\documentclass[twocolumn,journal]{IEEEtran}
\usepackage{verbatim}
\usepackage{amsmath}
\usepackage{amsthm}
\usepackage{amssymb}
\usepackage{color}
\usepackage{bm}
\usepackage{multirow}

\usepackage[pdftex]{graphicx}

\newtheorem{theorem}{Theorem}

\newtheorem{remark}{Remark}
\newtheorem{lemma}{Lemma}

\newtheorem{definition}{Definition}

\def\QED{\mbox{\rule[0pt]{1.5ex}{1.5ex}}}

\def\endproof{\hspace*{\fill}~\QED\par\endtrivlist\unskip}

\renewcommand{\qed}{\hfill \QED}

\def\FF{\mathbb{F}}

\def\im{\mathop{\rm Im}}

\def\Label#1{\label{#1}\ [\ \text{#1}\ ]\ }
\def\Label{\label}

\begin{document}

\title{
Reduction Theorem for Secrecy 
over Linear Network Code for Active Attacks
}

\author{Masahito Hayashi, \IEEEmembership{Fellow, IEEE}, 
Masaki Owari, Go Kato, and Ning Cai \IEEEmembership{Fellow, IEEE}
\thanks{This works was supported in part by
the Japan Society of the Promotion of Science (JSPS) Grant-in-Aid for Scientific Research (B) Grant 16KT0017 and
for Scientific Research (A) Grant 17H01280 and
for Scientific Research (C) Grant 16K00014 and No. 17K05591, 
in part by the Okawa Research Grant, and 
in part by the Kayamori Foundation of Informational Science Advancement.
The material in this paper was presented in part at the 2017 IEEE International Symposium on Information Theory (ISIT 2017),   Aachen (Germany), 25-30 June 2017 \cite{HOKC}.}
\thanks{Masahito Hayashi is with the Graduate School of Mathematics, Nagoya University, Nagoya, 464-8602, Japan. 
He is also with 
Shenzhen Institute for Quantum Science and Engineering, Southern University of Science and Technology,
Shenzhen, 518055, China,
Center for Quantum Computing, Peng Cheng Laboratory, Shenzhen 518000, China,
and the Centre for Quantum Technologies, National University of Singapore, 3 Science Drive 2, 117542, Singapore
(e-mail:masahito@math.nagoya-u.ac.jp).
Masaki Owari is with Department of Computer Science, Faculty of Informatics, Shizuoka University, Japan 
(e-mail:masakiowari@inf.shizuoka.ac.jp).
Go Kato is with NTT Communication Science Laboratories, NTT Corporation, Japan
(e-mail:kato.go@lab.ntt.co.jp).
Ning Cai is with the School of Information Science and Technology, ShanghaiTech University, Middle Huaxia Road no 393,
Pudong, Shanghai  201210, China
(e-mail: ningcai@shanghaitech.edu.cn).} }

\markboth{M. Hayashi, M. Owari, G. Kato, and N. Cai: Reduction Theorem for Secrecy 
over Linear Network Code for Active Attacks}{}

\maketitle

\begin{abstract}
We discuss the effect of sequential error injection on information leakage
under a network code.
We formulate a network code for the single transmission setting and
the multiple transmission setting.
Under this formulation, we show that 
the eavesdropper cannot improve the power of eavesdropping
by sequential error injection
when the operations in the network are linear operations.
We demonstrate the usefulness of this reduction theorem
by applying a concrete example of network.
\end{abstract}

\begin{IEEEkeywords} 
secrecy analysis,
secure network coding,
sequential injection,
passive attack,
active attack
\end{IEEEkeywords}

\section{Introduction}

Secure network coding offers a method for securely transmitting information from 
an authorized sender to an authorized receiver.
Cai and Yeung \cite{Cai2002} discussed the secrecy 
when the malicious adversary, Eve, wiretaps  a subset $E_E$ of the set $E$ of all the channels in a network.
Using the universal hashing lemma \cite{bennett95privacy,HILL,hayashi11}, 
the papers \cite{Matsumoto2011,Matsumoto2011a} showed the existence of a secrecy code that works universally for any type of eavesdropper 
when the cardinality of $E_E$ is bounded.
In addition, the paper \cite{KMU} discussed the construction of such a code.
As another type of attack on information transmission via a network,
a malicious adversary contaminates the communication by 
changing the information on a subset $E_A$ of $E$.
Using an error correction, the papers \cite{Cai06a,Cai06,HLKMEK,JLHE} proposed a method to protect the message from contamination.
That is, we require that the authorized receiver correctly recovers the message, which is called robustness.

As another possibility, we consider the case when the malicious adversary combines eavesdropping and contamination.
That is, contaminating a part of the channels, the malicious adversary might improve the ability of eavesdropping
while a parallel network offers no such a possibility \cite{KSZBJ,ZKBJS1,ZKBJS2}.
In fact, in arbitrarily varying channel model,
noise injection is allowed after Eve's eavesdropping, but  
Eve does not eavesdrop the channel after Eve's noise injection \cite{BBT,Ahlswede,CN88,TJBK}\cite[Table I]{KJL}.
The paper \cite{Yao2014} also discusses secrecy in the same setting
while it addresses the network model.
The studies \cite{KMU,Zhang} discussed the secrecy
when Eve eavesdrops the information transmitted on the channels in $E_E$ after noises are injected in $E_A$,
but they assume that Eve does not know the information of the injected noise. 
The paper \cite{Yao2014} discusses secrecy only for a passive attack. 

In contrast, this paper focuses on network, 
and discusses the secrecy when Eve 
adds artificial information to the information transmitted on the channels in $E_A$, 
eavesdrops the information transmitted on the channels in $E_E$, and 
estimates the original message from the eavesdropped information and the information of the injected noises.
We call this type of attack an {\it active attack} 
and call an attack without contamination a {\it passive attack}.
Specially, we call each of Eve's active operations a strategy.
When $E_A \subset E_E$ and any active attack is available for Eve, she is allowed to arbitrarily modify the information on the channels in $E_A$ sequentially based on the obtained information.

This paper aims to show a reduction theorem for an active attack, i.e.,
the fact that no strategy can improve Eve's information 
when every operation in the network is linear 
and Eve's contamination satisfies a natural causal condition.
When the network is not well synchronized, Eve can make an attack across several channels.
This reduction theorem holds even under this kind of attack.
In fact, there is an example having a non-linear node operation such that 
Eve can improve her performance to extract information from eavesdropping an edge outgoing an intermediate node
by adding artificial information to an edge incoming the intermediate node
\cite{HOKC2}.
This example shows the necessity of linearity for this reduction theorem.
Although our discussion can be extended to 
the multicast and multiple-unicast cases,
for simplicity, we consider the unicast setting in the following discussion.

Further, we apply our general result to the analysis of a concrete example of a network.
In this network, we demonstrate that  any active attack cannot improve the performance of eavesdropping.
However, in the single transmission case over the finite field $\FF_2$,
the error correction and the error detection is impossible over this contamination.
To resolve this problem, 
this paper addresses the multiple transmission case in addition to the single transmission case.
In the multiple transmission case, the sender uses the same network multiple times,
and the topology and dynamics of the network do not change during these transmissions.
While several papers discussed this model, many of them discussed 
the multiple transmission case only with contamination \cite{JLKHKM,Jaggi2008,JL} or eavesdropping \cite{Matsumoto2011,Matsumoto2011a}.
Only the paper \cite{Yao2014} addressed it with contamination and eavesdropping, i.e., it assumed that all contaminations are done after eavesdropping.
We formulate the multiple transmission case when each transmission has no correlation with the previous transmission
while injected noise might have such a correlation.
Then, we show the above type of reduction theorem for an active attack even under
the multiple transmission case.
We apply this result to the multiple transmission over the above example of a network,
in which, the error correction and the error detection are possible over this contamination.
Hence, the secrecy and the correctness hold in this case.

The remaining part of this paper is organized as follows.
Section \ref{S2-1} discusses only the single transmission setting that has only a single transmission
and 
Section \ref{S2-5} does the multiple transmission setting that has $n$ transmissions.
Two types of multiple transmission settings are formulated.
Then, we state our reduction theorem in both settings.
In Section \ref{SCon}, we state the conclusion.


\section{Single transmission setting}\Label{S2-1}
\subsection{Generic model}
In this subsection, we give a generic model, and discuss its relation with a concrete network model in the latter subsections.
We consider the unicast setting of network coding on a network.
Assume that the authorized sender, Alice, intends to send information to the authorized receiver, Bob,
via the network.
Although the network is composed of $m_1$ edges and $m_2$ vertecies,
as shown in later, the model can be simplified as follows when the node operations are linear. 
We assume that
Alice inputs the input variable ${\bf X}$ in $\FF_q^{m_{{3}}}$ and Bob receives the output variable ${\bf Y}_B$ in $\FF_q^{m_{{4}}}$,
where $\FF_q$ is a finite field whose order is a power $q$ of the prime $p$.
We also assume that the malicious adversary, Eve, wiretaps the information ${\bf Y}_E$ in $\FF_q^{m_{6}}$\footnote{In this paper, we denote the vector on $\FF_q$ by a bold letter.
But, we use a non-bold letter to describe a scalar and a matrix.}. 
Then, we adopt the model with matrices $K_B \in \FF_q^{m_{{4}}\times m_{{3}} }$ and $K_E \in \FF_q^{m_{6}\times m_{{3}} }$, in which,the variables ${\bf X}$, ${\bf Y}_B$, and ${\bf Y}_E$ satisfy their relations
\begin{align}
{\bf Y}_B=K_B {\bf X}, \quad
{\bf Y}_E=K_E {\bf X}.\Label{F2}
\end{align}
This attack is a conventional wiretap model and is called 
a {\it passive attack} to distinguish an active attack, which will be introduced later.
Section \ref{PDG} will explain how
this model is derived from a directed graph with $E_E$
and linear operations on nodes.

In this paper, we address a stronger attack, in which, Eve injects noise ${\bf Z} \in \FF_q^{m_{5}}$.
Hence, using matrices $H_B \in \FF_q^{m_{{4}}\times m_{5} }$ and $H_E \in \FF_q^{m_{6}\times m_{5}}$, 
we rewrite the relations \eqref{F2} as
  \begin{align}
{\bf Y}_B=K_B {\bf X}+ H_B {\bf Z}, \quad
{\bf Y}_E=K_E {\bf X}+ H_E {\bf Z}, 
\Label{E3}
\end{align}
which is called a {\it wiretap and addition model}.
The $i$-th injected noise $Z_i$ (the $i$-th component of ${\bf Z}$) is decided by a function 
$\alpha_i$ of ${\bf Y}_E$.
Although a part of ${\bf Y}_E$ is a function of $\alpha_i$,
this point does not make a problem for causality, as explained in Section \ref{PDG2-5}.
In this paper, when a vector has the $j$-th component $x_j$,
the vector is written as $[x_j]_{1\le j \le a}$, where
the subscript $1\le j \le a$ expresses the range of the index $j$.
Thus, the set $\alpha=[\alpha_i]_{1\le i\le m_5}$ of the functions can be regarded as Eve's strategy,
and we call this attack an {\it active attack with a strategy} $\alpha$.
That is, an active attack is identified by a pair of a strategy $\alpha$ 
and a wiretap and addition model decided by $\bm{K}, \bm{H}$.
Here, we treat $K_B,K_E,H_B$, and $H_E$ as deterministic values, and denote the pairs $(K_B,K_E)$ and 
$(H_B,H_E)$ by $\bm{K}$ and $\bm{H}$, respectively.
Hence, our model is written as the triplet $(\bm{K}, \bm{H}, \alpha )$.
As shown in the latter subsections, 
under the linearity assumption on the node operations,
the triplet $(\bm{K}, \bm{H}, \alpha )$ is decided from the network topology (a directed graph with $E_A$ and $E_E$)
and dynamics of the network.
Here, we should remark that 
the relation \eqref{E3} is based on the linearity assumption for node operations. 
Since this assumption is the restriction for the protocol, it does not restrict the eavesdropper's strategy.


\begin{table}[htpb]
  \caption{Channel parameters}
\Label{hikaku}
\begin{center}
  \begin{tabular}{|c|l|} 
\hline
$m_1$ & Number of edges\\
\hline
$m_2$ & Number of vertecies\\
\hline
$m_3$ & Dimension of Alice's input information ${\bf X}$ \\
\hline
$m_4$ & Dimension of Bob's observed information ${\bf Y}_B$ \\
\hline
$m_5$ & Dimension of Eve's injected information ${\bf Z}$\\
\hline
$m_6$ & Dimension of Eve's wiretapped information ${\bf Y}_E$\\
\hline
$m_7$ & $m_1-m_3$ \\
\hline
  \end{tabular}
\end{center}
\end{table}

We impose several types for regularity conditions for Eve's strategy $\alpha$, which are demanded from causality.
Notice that $\alpha_i $ is a function of the vector $[Y_{E,j}]_{1\le j \le m_6}$.
Now, we take the causality with respect to $\alpha$ into account.
Here, we assume that the assigned index $i$ for $1\le i \le m_5$ expresses the time-ordering of injection.
That is, we assign the index $i$ for $1\le i \le m_5$ according to the order of injections.
Hence, we assume that $\alpha_i $ is decided by a part of Eve's observed variables. 
We say that subsets $w_i\subset \{1, \ldots, m_6\}$ for $i \in \{1, \ldots, m_5\}$ 
are the {\it domain index subsets for} $\alpha$
when the function $\alpha_i $ is given as a function of the vector $[Y_{E,j}]_{j \in w_i}$.
Here, the notation $j \in w_i$ means that
the $j$-th eavesdropping is done before the $i$-th injection, i.e.,
$w_i$ expresses the set of indexes corresponding to the symbols 
that do effect the $i$-th injection.
Hence, the eavesdropped symbol $Y_{E,j}$ does not depend on the injected symbol $z_i$ for $j \in w_i$.
Since the decision of the
injected noise does not depend on the consequences of the decision, 
we introduce the following causal condition.

\begin{definition}\Label{Def3}
We say that 
the domain index subsets $\{w_i\}_{1, \ldots, m_5}$ satisfy the {\it causal condition}
when the following two conditions hold;
\begin{description}
\item[(A1)] The relation $ H_{E;j,i} = 0$ holds for $j \in w_i$.
\item[(A2)] 
The relation $ w_{1} \subseteq w_{2} \subseteq \ldots \subseteq 
w_{m_5}$ holds.
\end{description}
\end{definition}


As a necessary condition of the causal condition,
we introduce the following uniqueness condition for the function $\alpha_i $, which is given as a function of the vector $[Y_{E,j}]_{1\le j \le m_6}$.

\begin{definition}
For any value of $ {\bf x}$,
there uniquely exists ${\bf y} \in \FF_q^{m_6} $ such that
\begin{align}
{\bf y}= K_E {\bf x}+ H_E \alpha({\bf y}).\Label{Uni}
\end{align}
This condition is called the {\it uniqueness condition for} $\alpha$.
\end{definition}

Examples of a network with $w_i$, $[H_{E;j,i}]_{i,j}$
will be given in Subsection \ref{F1Ex}.
Then, we have the following lemma.

\begin{lemma}\Label{LL2}
When 
a strategy $\alpha$ has domain index subsets to satisfy the causal condition, 
the strategy $\alpha$ satisfies the uniqueness condition.
\end{lemma}

\begin{proof}
When the causal condition holds,
we show the fact that $y_{j'}$ is given as a function of $K_E {\bf x}$ for any $j' \in w_{i}$
by induction with respect to the index $i =1, \ldots, m_5$, which expresses the order of the injected information.
This fact yields the uniqueness condition.

For $j \in w_1$, we have $y_j= (K_E {\bf x})_j$ because $(H_E \alpha(y))_{j}$ is zero.
Hence, the statement with $i=1$ holds.
We choose $j \in w_{i+1}\setminus w_{i}$.
Let $z_{i'}$ be the $i'$-th injected information.
Due to Conditions (A1) and (A2),
$y_j- (K_E {\bf x})_j=(H_Ez)_j$ is a function of 
$z_1=a(y)_1,\cdots,z_i=a(y)_i $.
Since the assumption of the induction guarantees that
$z_{1}, \ldots, z_{i}$ are functions of 
$[y_{j'}]_{j' \in w_{i}}$,
$z_{1}, \ldots, z_{i}$ are functions of $K_E {\bf x}$.
Then, we find that $y_j=(K_Ex)_j+(H_Ez)_j$ is given as a function of 
$K_E {\bf x}$ for any 
$j \in w_{i+1}\setminus w_{i}$.
That is, the strategy $\alpha$ satisfies the uniqueness condition.
\end{proof}

Now, we have the following reduction theorem.
\begin{theorem}[Reduction Theorem]\Label{T0}
When the strategy $\alpha$ satisfies the uniqueness condition, 
Eve's information ${\bf Y}_E(\alpha)$ with strategy $\alpha$ can be calculated from 
Eve's information ${\bf Y}_E(0)$ with strategy $0$ (the passive attack),
and ${\bf Y}_E(0)$ is also calculated from ${\bf Y}_E(\alpha)$.
Hence, we have the equation
\begin{align}
I({\bf X};{\bf Y}_E)[0]
=&
I({\bf X};{\bf Y}_E)[\alpha], \Label{NCT1}
\end{align}
$I({\bf X};{\bf Y}_E)[\alpha]$ expresses the mutual information
between ${\bf X}$ and ${\bf Y}_E$ under the strategy $\alpha$.
\end{theorem}

\begin{proof}
Since ${\bf Y}_E(0)=K_E {\bf X}$ and 
${\bf Y}_E(\alpha)=K_E {\bf X}+ H_E {\bf Z}$,
due to the uniqueness condition of the strategy $\alpha$, 
we can uniquely evaluate ${\bf Y}_E(\alpha)$
from ${\bf Y}_E(0)=K_E {\bf X}$ and $\alpha$.
Therefore, we have 
$I({\bf X};{\bf Y}_E)[0]\ge I({\bf X};{\bf Y}_E)[\alpha]$.
Conversely, since ${\bf Y}_E(0)$ is given as a function 
($Y_E(\alpha)-H_E {\bf Z}$)
of ${\bf Y}_E(\alpha)$, ${\bf Z}$, and $H_E$,
we have the opposite inequality.
\end{proof}

This theorem shows that the information leakage of the active attack with 
the strategy $\alpha$ is the same as 
the information leakage of the passive attack.
Hence, to guarantee the secrecy under an arbitrary active attack,
it is sufficient to show secrecy under the passive attack.
However, there is an example of non-linear network such that
this kind of reduction does not hold \cite{HOKC2}. 
In fact, even when the network does not have synchronization so that the
information transmission on an edges starts before the end of
the information transmission on the previous edge,
the above reduction theorem hold under the uniqueness condition.

\subsection{Construction of $K_B,K_E$ 
from concrete network model}\Label{PDG}
Next, we discuss how we can obtain the generic passive attack model \eqref{F2} 
from a concretely structured network coding, i.e. communications identified by directed edges and linear operations by parties identified by nodes.
We consider the unicast setting of network coding on a network, which is given 
as a directed graph $({V}, {E})$, 
where the set ${V}:=\{v(1), \ldots, v(m_2)\}$ of vertices expresses the set of nodes
and the set ${E}:=\{e(1), \ldots, e(m_1)\}$ of edges expresses the set of communication channels, where 
a communication channel means a packet in network engineering,
i.e., a single communication channel can transmit single character in 
$\FF_q$.
In the following, we identify 
the set ${E}$ with $\{1, \ldots, m_1\}$, i.e,
we identify the index of an edge with the edge itself.
Here, the directed graph $({V}, {E})$ is not necessarily acyclic.
When a channel transmits information from a node $v(i)\in {V}$ to another node $v(i')\in {V}$, it is written as $(v(i),v(i')) \in {E}$.

In the single transmission,
the source node has several elements of $\FF_q$
and sends each of them via its outgoing edges 
in the order of assigned number of edges.
Each intermediate node keeps received information via incoming edges.
Then, for each outgoing edge,
the intermediate node calculates one element of $\FF_q$ 
from previously received information,
and sends it via the outgoing edge.
That is, every outgoing information from a node $v(i)$ via a channel $e(j)$ 
depends only on the incoming information into the node $v(i)$ via channels 
$e(j')$ such that $j'<j$.
The operations on all nodes are assumed to be linear on the finite field $\FF_q$ with prime power $q$.
Bob receives the information ${\bf Y}_B$ in $\FF_q^{m_{4}}$ 
on the edges of a subset $E_B:=\{e(\zeta_B(1)), \ldots, e(\zeta_B(m_{4}))\}\subset E$,
where $\zeta_B$ is a strictly increasing function from $\{1, \ldots, m_4\}$ to $\{1, \ldots, m_1\}$.
Let $\tilde{X}_j$ be the information on the edge $e(j)$.
In the following, we describe the information on the $m_7:=m_1-m_3$ edges that are not directly linked to the source node
because $m_3$ expresses the number of Alice’s input symbols.
When the edge $e(j)$ is an outgoing edge of the node $v(i)$,
the information $\tilde{X}_{j}$ is given as a linear combination of 
the information on the edges incoming to the node $v(i)$.
We choose an $m_1 \times m_1$ matrix $\theta=(\theta_{j,j'})$ such that
$\tilde{X}_{j}= \sum_{j'}\theta_{j,j'} \tilde{X}_{j'}$, where
$\theta_{j,j'} $ is zero 
unless $ e(j')$ is an edge incoming to $v(i)$.
The matrix $\theta$ is the {\it coefficient matrix} of this network.

Now, from causality, we can
assume that each node makes the transmissions on the outgoing edges 
in the order of the numbers assigned to the edges.
At the first stage,
all $m_3$ information generated at the source node are directly transmitted via $e(1),\cdots e(m_3)$ respectively.
Then, at time $j$, the information transmission on the edge $e(j+m_3)$ is done 
Hence, naturally, we impose the condition
\begin{align}
\theta_{j,j'}=0 \hbox{ for } j' \ge j \Label{GTY},
\end{align}
which is called the {\it partial time-ordered condition for $\theta$.} 
Then, to describe the information on $m_7$ edges that are not directly linked to the source node,
we define $m_7$ $m_1 \times m_1$ matrices
$M_1, \ldots, M_{m_7}$.
The $j$-th $m_1 \times m_1$ matrix $M_j$ 
gives the information on the edge $e(j+m_3)$ 
as a function of the information on edges $\{e(j')\}_{1\le j'\le m_1}$
at time $j$.
The $j+m_3$-th row vector of the matrix $M_j$ is defined by 
$[\theta_{j+m_3,j'}]_{1\le j'\le m_1}$.
The remaining part of $M_j$, i.e.,
the $i$-th row vector for $i \neq j+m_3$
is defined by $[\delta_{i,j'}]_{1\le j'\le m_1}$
and $\delta_{i,j'}$ is the Kronecker delta.
Since $\sum_{i=1}^{m_3}(M_{j}\cdots M_1)_{j',i} X_i$
expresses the information on edge $e(j')$ at time $j$,
we have
\begin{align}
Y_{B,j}= \sum_{i=1}^{m_3}(M_{m_7}\cdots M_1)_{\zeta_B(j),i} X_i\Label{KOG}
\end{align}
While the output of the matrix $M_{m_7}\cdots M_1$ takes values in $\FF_q^{m_1}$,
we focus the projection $P_B$ to the subspace $\FF_q^{m_4}$ that corresponds to the $m_4$ components observed by Bob.
That is, $P_B$ is a $m_4 \times m_1$ matrix to satisfy $P_{B;i,j}=\delta_{\zeta_B(i),j}$.
Similarly, 
we use the projection $P_A$ (an $m_1 \times m_3$ matrix) as $P_{A;i,j}=\delta_{i,j}$.
Due to \eqref{KOG}, the matrix $K_B:=P_B  M_{m_7}\cdots M_1 P_A$ satisfies the first equation in \eqref{F2}.

The malicious adversary, Eve, wiretaps the information $Y_E$ in $\FF_q^{m_{6}}$ 
on the edges of a subset $E_E:=\{e(\zeta_E(1)), \ldots, e(\zeta_E(m_{6}))\}\subset E$,
where $\zeta_E$ is a strictly increasing function from 
$\{1, \ldots, m_6\}$ to $\{1, \ldots, m_1\}$.
Similar to \eqref{KOG}, we have
\begin{align}
Y_{E,j}= \sum_{i=1}^{m_3}(M_{m_7}\cdots M_1)_{\zeta_E(j),i} X_i\Label{KOG2}.
\end{align}
We employ the projection $P_E$ (an $m_6 \times m_1$ matrix) to the subspace $\FF_q^{m_6}$ that corresponds to the $m_6$ components eavesdropped by Eve.
That is, $P_{E;i,j}= \delta_{\zeta_E(i),j}$.
Then, we obtain the matrix $K_E$ as $P_E  M_{m_7}\cdots M_1 P_A$.
Due to \eqref{KOG}, the matrix $K_E:=P_E  M_{m_7}\cdots M_1 P_A$ satisfies the second equation in \eqref{F2}.

In summary 
the topology and dynamics (operations on the intermediate nodes) of the network, including the places of attached edges decides
the graph $(V,E)$, the coefficients $\theta_{i,j}$, and functions $\zeta_B,\zeta_E$,
uniquely gives the two matrices $K_B$ and $K_E$.
Subsection \ref{F1Ex} will give an example for this model.
Here, we emphasize that we do not assume the acyclic condition for the graph $({V}, {E})$.
We can use this relaxed condition because we have only one transmission in the current discussion.
That is, due to the partial time-ordered condition for $\theta$,
we can uniquely define our matrices $K_B$ and $K_E$, which is a similar way to \cite[Section V-B]{ACLY}\footnote{%
$\Lambda$ of Ahlswede-Cai-Li-Yeung corresponds to
the number of edges that are not connected to the source node in our paper.}.
However, when the graph has a cycle and we have $n$ transmissions,
there is a possibility of the correlation with the delayed information
dependently of the time ordering.
As a result, it is difficult to analyze secrecy for the cyclic network coding. 
 
\subsection{Construction of $H_B,H_E$ 
from concrete network model
}\Label{PDG2}
We identify the 
wiretap and addition model
from a concrete network structure.
We assume that Eve injects the noise in a part of edges $E_A \subset E$
as well as eavesdrops the edges $E_E$.

The elements of the subset $E_A$ are expressed as
$E_A=\{e(\eta(1)), \ldots, e(\eta(m_{5}))\}$ 
by using a function $\eta$ from 
$\{1, \ldots, m_5\}$ to $\{1, \ldots, m_1\}$,
where the function $\eta$ is not necessarily monotonically increasing function.
To give the matrices $H_B$ and $H_E$,
modifying the matrix $M_j$,
we define the new matrix $M_j'$ as follows
The $j+m_3$-th row vector of the new matrix $M_j'$ is defined by 
$[\theta_{j+m_3,j'}+\delta_{j+m_3,j'}]_{1\le j'\le m_1}$.
The remaining part of $M_j'$, i.e.,
the $i$-th row vector for $i \neq j+m_3$
is defined by $[\delta_{i,j'}]_{1\le j'\le m_1}$.
Since 
$\sum_{i=1}^{m_3}(M_{j}\cdots M_1)_{j',i} X_i
+\sum_{i'=1}^{m_5}(M_{j}'\cdots M_1')_{j',\eta(i')} Z_{i'} $
expresses the information on edge $e(j')$ at time $j$,
we have
\begin{align}
Y_{B,j}= &
\sum_{i=1}^{m_3}(M_{m_7}\cdots M_1)_{\zeta_B(j),i} X_i
\nonumber \\
&+
\sum_{i'=1}^{m_5}(M_{m_7}'\cdots M_1')_{\zeta_B(j),\eta(i')} Z_{i'} 
\Label{KOG3}\\
Y_{E,j}= &
\sum_{i=1}^{m_3}(M_{m_7}\cdots M_1)_{\zeta_E(j),i} X_i
\nonumber \\
&+\sum_{i'=1}^{m_5}(M_{m_7}'\cdots M_1'- I)_{\zeta_E(j),\eta(i')} Z_{i'}.
\Label{KOG4}
\end{align}
When Eve eavesdrops the edges $E_E \cap E_A$,
she obtains the information on $E_E\cap E_A$ before her noise injection.
Hence, to express her obtained information on $E_E\cap E_A$,
we need to subtract  her injected information on $E_E\cap E_A$.
Hence, we need $-I$ in the second term of \eqref{KOG4}.
We introduce the projection $P_{E,A}$ (an $m_1 \times m_5$ matrix) as 
$P_{E,A;i,j}=\delta_{i,\eta(j)}$.
Due to \eqref{KOG3} and \eqref{KOG4}, 
the matrices 
$H_B:=P_B  M_{m_7}'\cdots M_1' P_{E,A}$ 
and
$H_E:=P_E  (M_{m_7}'\cdots M_1'-I) P_{E,A}$ satisfy
conditions \eqref{E3} with the matrices
$K_B$ and $K_E$, respectively.
This model ($K_B$, $K_E$, $H_B$, $H_E$) to give \eqref{E3}
is called the {\it wiretap and addition model}
determined by $(V,E)$ and $(E_E, E_A,\theta)$, which expresses 
the topology and dynamics.

\subsection{Strategy and order of communication}\Label{PDG2-5}
To discuss the active attack, we see how the causal condition for 
the subsets $\{w_i\}_{1, \ldots, m_5}$ 
follows from the network topology in the wiretap and addition model.
We choose the domain index subsets $\{w_i \}_{1\le i \le m_5}$ for $\alpha$,
i.e., Eve chooses the added error $Z_i$ on the edge $e(\eta(i)) \in E_A$
as a function $\alpha_i $ of the vector $[Y_{E,j}]_{j \in w_i}$.
Since the order of Eve's attack is characterized by
the function $\eta$ from $\{1, \ldots, m_5\}$ to $E_A \subset \{1, \ldots, m_1\}$,
we discuss what condition for the pair $(\eta,\{w_i\}_i)$ guarantees 
the causal condition for the subsets $\{w_i\}_i$. 

First, 
one may assume that the tail node of 
the edge $e(j)$ sends the information to the edge $e(j)$
after the head node of the edge $e(j-1)$ receives the information to the edge $e(j-1)$.
Since this condition determines the order of Eve's attack, 
the function $\eta$ must be a strictly increasing function from 
$\{1, \ldots, m_5\}$ to $\{1, \ldots, m_1\}$.
Also, due to this time ordering, the subset $w_i$ needs to be $\{j| \eta(i) \ge \zeta_E(j)\}$ or its subset.
We call these two conditions 
the {\it full time-ordered condition for the function $\eta$ and the subsets $\{w_i\}_i$}.
Since the function $\eta$ is strictly increasing, Condition (A2) for the causal condition holds.
Since the relation \eqref{GTY} implies that 
$M_{m_7}'\cdots M_1'- I$ is a lower triangular matrix with zero diagonal elements,
the strictly increasing property of $\eta$ yield that
\begin{align}
H_{E;j,i}=0 \hbox{ when } \eta(i) \ge \zeta_E(j) \Label{F14},
\end{align}
which implies Condition (A1) for the causal condition.
In this way, the full time-ordered condition for the function $\eta$ and the subsets $\{w_i\}_i$ satisfies the causal condition.

However, the full time ordered condition does not hold in general 
even when we reorder the numbers assigned to the edges.
That is, 
if the network is not well synchronized, Eve can make an attack across several channels, i.e.,
it is possible that Eve might intercept (i.e., wiretap and contaminate) the information of an edge
before the head node of the previous edge receives the information on the edge.
Hence, we consider the case when the partial time-ordered condition holds, but
the full time-ordered condition does not necessarily hold\footnote{%
For an example, we consider the following case.
Eve gets the information on the first edge. Then, she gets 
the information on the second edge before she hands over 
the information on the first edge to the tail node of the first edge.
In this case, she can change the information on the first edge
based on the information on the first and second edges. Then, the time-ordered condition \eqref{F14} does not hold.}.
That is, the function $\eta$ from $\{1, \ldots, m_5\}$ to $E $ is injective but is not necessarily monotone increasing.
Given the matrix $\theta$,
we define the function 
$\gamma_\theta(j):=\min_{j'} \{ j'|   \theta_{j',j}\neq 0\}$.
Here, when no index $j'$ satisfies the condition $\theta_{j',j}\neq 0$,
$\gamma_\theta(j)$ is defined to be $m_1+1$.
Then, we say that the function $\eta$ and the subsets $\{w_i\}_i$ are {\it admissible under $\theta$} 
when $\{e(k) | k \in \im \eta \}=E_A$,
the subsets $\{w_i\}_i$ satisfy Condition (A2) for the causal condition, 
and any element $j \in w_{i}$ satisfies
\begin{align}
\zeta_E(j) < \gamma_\theta(\eta(i)).
\Label{LGY}
\end{align}
Here, $\im \eta$ expresses the image of the function $\eta$.
The condition \eqref{LGY} and the condition \eqref{GTY} imply the following condition; 
For $ j \in w_i$,
there is no sequence $\zeta_E(j)=j_1>j_2, \ldots >j_l=\eta(i) $ 
such that 
\begin{align}
\theta_{j_{i},j_{i+1}}\neq 0 .
\end{align}
This condition implies Condition (A1) for the causal condition.  
Since the admissibility under $\theta$ is natural,
even when the full time-ordered condition does not hold,
the causal condition can be naturally derived.

Given two admissible pairs $(\eta,\{w_i\}_i)$ and $(\eta',\{w_i'\}_i)$, 
we say that the pair $(\eta,\{w_i\}_i)$ is {\it superior to} 
$(\eta',\{w_i'\}_i)$ for Eve
when $w_{{\eta'}^{-1}(j)}'\subset w_{\eta^{-1}(j)} $ for 
any $j \in E_A$.
Now, we discuss the optimal choice of $(\eta, \{w_i\}_i)$ in this sense
when $E_A$ is given.
That is, we choose the subset $w_i$ as large as possible under 
the admissibility under $\theta$.
Then, we choose the bijective function $\eta_{o}$ from $\{1, \ldots, m_5\}$
to $E_A$ 
such that $\gamma_\theta \circ \eta_o$ is monotone increasing.
Then, we define $w_{o,i}:=\{ j| \zeta_E(j) < \gamma_\theta(\eta_o(i))\}$,
which satisfies the admissibility under $\theta$.
Conditions (A1) and (A2) for the causal condition.  
Further, when the pair $( \eta,\{w_i\}_i)$ is admissible under $\theta$,
the condition \eqref{LGY} implies 
$w_{\eta^{-1} (j)}\subset w_{o,\eta_o^{-1}(j)} $ for
$j \in E_A$, i.e., 
$w_{o,i} $ is the largest subset under the admissibility under $\theta$.
Hence, we obtain the optimality of $(\eta_o, \{w_{o,i}\}_i)$.
Although the choice of $\eta_o$ is not unique, 
the choice of $w_{o,\eta_o^{-1}(j)}$ for $j \in E_A$
 is unique.

\subsection{Secrecy in concrete network model}\Label{F1Ex}

\begin{figure}[htbp]
\begin{center}
\includegraphics[scale=0.4, angle=-90]{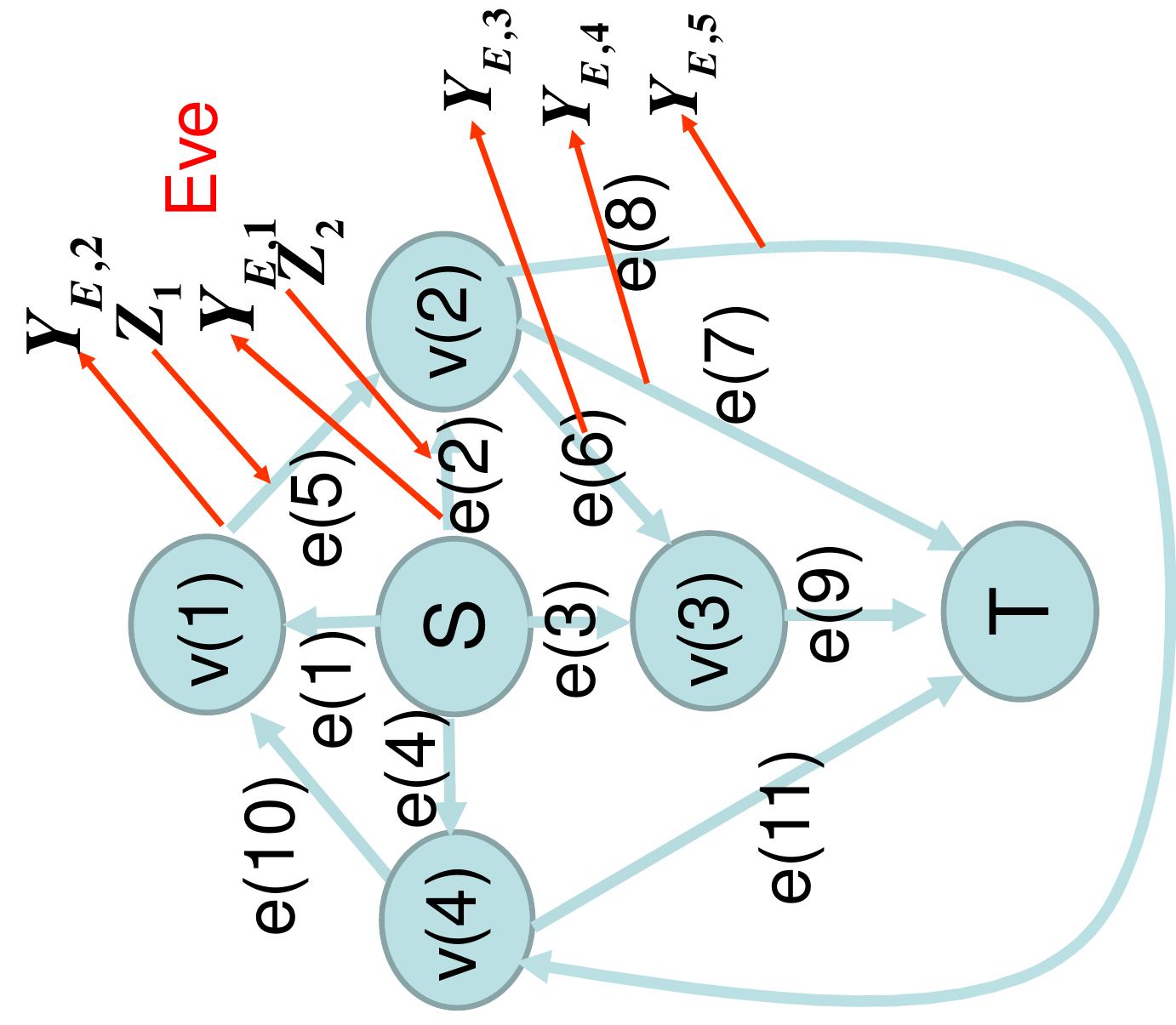}
\end{center}
\caption{Network of Subsection \ref{F1Ex} with name of edges}
\Label{F3}
\end{figure}%

\begin{figure}[htbp]
\begin{center}
\includegraphics[scale=0.4, angle=-90]{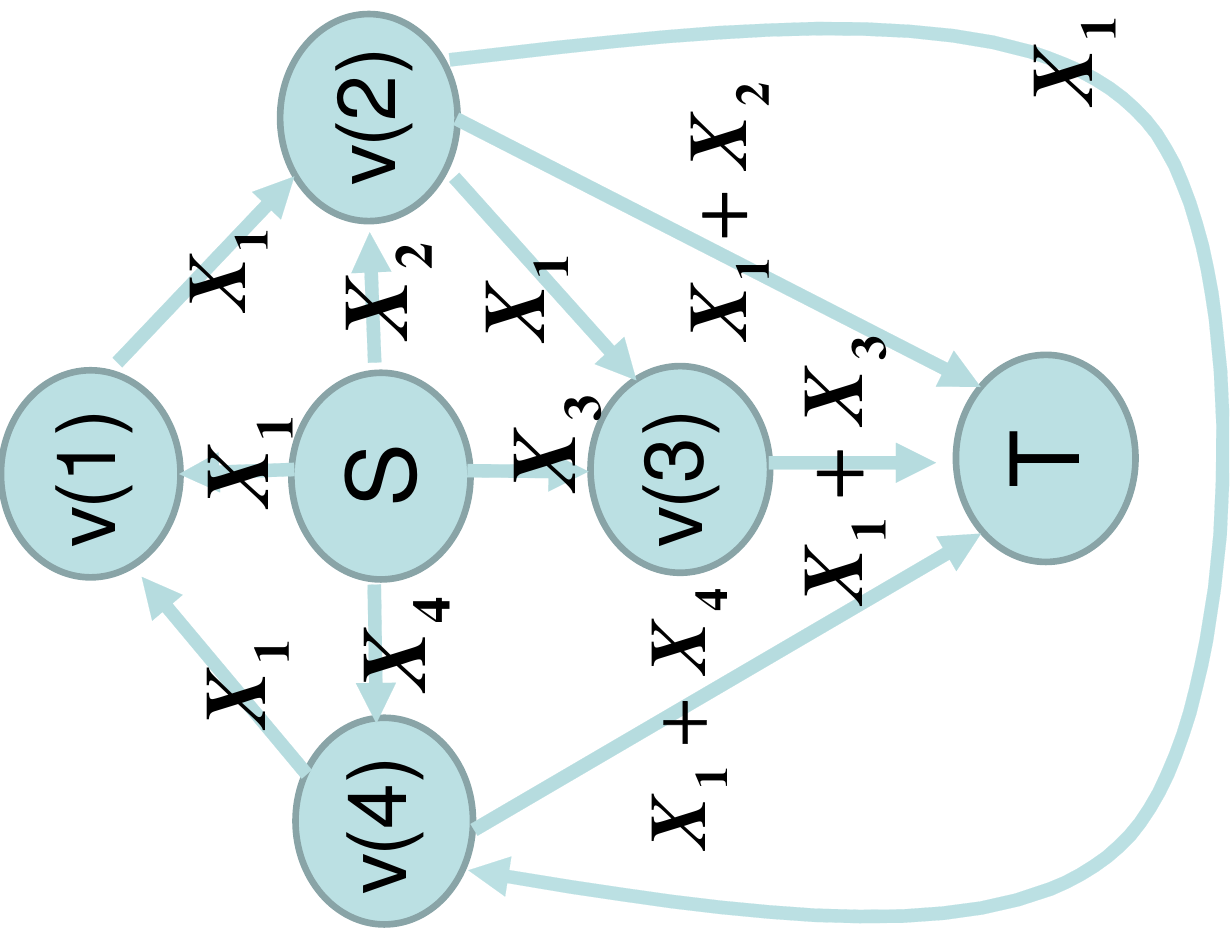}
\end{center}
\caption{Network of Subsection \ref{F1Ex} with network flow}
\Label{F1B}
\end{figure}%

In this subsection, as an example, we consider the network given in Figs. \ref{F3} and \ref{F1B}, which shows that
our framework can be applied to the network without synchronization.
Alice sends the variables $X_1,\ldots, X_4 \in \FF_q$ 
to nodes $v(1),v(2),v(3),$ and $v(4)$
via the edges
$e(1),e(2),e(3)$, and $e(4)$, respectively.
The edges $e(5),e(6),e(8), e(10)$
send the elements received from the edges 
$e(1), e(5),e(5),e(8)$, respectively.
The edges $e(7)$, $e(9)$, and $e(11)$ 
send the sum of two elements received from the edge pairs
$(e(2),e(5)),$ $(e(3),e(6)),$ and $(e(4),e(8)),$ respectively.

Bob received elements via the edges 
$e(7), e(9), e(11),$ which are written as $Y_{B,1}, Y_{B,2}, Y_{B,3}$, respectively.
Then, the matrix $K_B$ is given as 
\begin{align}
K_B= 
\left(
\begin{array}{cccc}
1 & 1 & 0 & 0 \\
1 & 0 & 1 & 0 \\
1 & 0 & 0 & 1 
\end{array}
\right).
\end{align}
Then, $m_3=4$ and $m_4=3$.

Now, we assume that Eve eavesdrops 
the edges $e(2),e(5), e(6),e(7),e(8)$, i.e., all edges connected to $v(2)$,
and contaminates
the edge $e(2),e(5)$.
Then, we set 
$\zeta_B(1)=7,\zeta_B(2)=9,\zeta_B(3)=11$
and 
$\zeta_E(1)=2,\zeta_E(2)=5,\zeta_E(3)=6,\zeta_E(4)=7,\zeta_E(5)=8$.
Eve can choose the function $\eta$ as 
\begin{align}
\eta(1)=5,
\eta(2)=2 \Label{eta-c}
\end{align}
while $\eta(1)=2, \eta(2)=5$ is possible.
In the following, we choose \eqref{eta-c}. 
Since $\gamma_\theta(2)=7$ and 
$\gamma_\theta(5)=6$,
the subsets $w_i$ are given as
\begin{align}
w_1:=w_{o,1}=\{ 1,2 \}, \quad
w_2:=w_{o,2}=\{ 1,2,3 \}
\end{align}
This case satisfies Conditions (A1) and (A2).
Hence, this model satisfies the causal condition.
Lemma \ref{LL2} guarantees that any strategy also satisfies the uniqueness condition.

We denote the observed information 
on the edges $e(2),e(5), e(6),e(7),e(8)$
by $Y_{E,1}, Y_{E,2},Y_{E,3}, Y_{E,4}, Y_{E,5}$.
As Fig. \ref{F3}, Eve adds $Z_{1},Z_2$ in edges $e(2),e(5)$.
Then, the matrices $H_B$, $K_E$, and $H_E$ are given as 
\begin{align}
H_B&= 
\left(
\begin{array}{cc}
1&1 \\
0&0 \\
0&0
\end{array}
\right), \quad
K_E= 
\left(
\begin{array}{cccc}
 0 & 1 & 0 & 0 \\
 1 & 0 & 0 & 0\\
 1 &0 & 0 & 0\\
 1  &1 & 0 & 0\\
 1& 0& 0 & 0 
\end{array}
\right), 
\nonumber \\
H_E&= 
\left(
\begin{array}{cc}
0 & 0 \\
0 & 0 \\
0 & 1 \\
1 & 1 \\
0 & 1 
\end{array}
\right).
\end{align}

In this case, to keep the secrecy of the message to be transmitted,
Alice and Bob can use coding as follows.
When Alice's message is $M \in \FF_q$,
Alice prepares scramble random number $L_1, L_2, L_3\in \FF_q$.
These variables are assumed to be subject to the uniform distribution independently.
She encodes them as 
$X_i= L_i$ for $i=1,\ldots, 3$ and $X_4= -M+L_1+L_2+L_3$.
As shown in the following, under this code, Eve cannot obtain any information for $M$
even though she makes active attack.
Due to Theorem \ref{T0},
it is sufficient to show the secrecy when $Z_i=0$.
Eve's information is $Y_{E,1}=X_2,Y_{E,2}=X_1,Y_{E,3}=X_1,Y_{E,4}= X_1+X_2$ and $Y_{E,5}=X_1$
and the message is $M=X_1+X_2+X_3-X_4$.
That is, her eavesdropping information is characterized by the vectors
$(0,1,0,0)$, $(1,0,0,0)$, $(1,0,0,0)$ and $(1,1,0,0)$
and the message is by the vector $(1,1,1,-1)$.
Since these vectors are linearly independent,
$X_1+X_2+X_3-X_4$ is independent of each of the variables $Y_{E,1},Y_{E,2},Y_{E,3},Y_{E,4},Y_{E,5}$.
Hence, the message is independent of her eavesdropping information. 

Indeed, 
the above attack can be considered as the following.
Eve can eavesdrop all edges connected to the intermediate node $v(2)$
and contaminate all edges incoming to the intermediate node $v(2)$.
Hence, it is natural to assume that Eve similarly eavesdrops and contaminates at another intermediate node $v(i)$. 
That is, 
Eve can eavesdrop all edges connected to the intermediate node $v(i)$
and contaminate all edges incoming to the intermediate node $v(i)$.
For all node $v(i)$, this code has the same secrecy against the above  Eve's attack for node $v(i)$. 

Furthermore, the above code has the secrecy even when the following attack.
\begin{description}
\item[(B1)]
Eve eavesdrops one of three edges $e(7),e(9),e(11)$ connected to the sink node,
and eavesdrops and contaminates one of the remaining eight edges
$e(1),e(2),e(3),e(4),e(5),e(6),e(8),e(10)$ that are not connected to the sink node.
\end{description}
Indeed, 
the vector characterizing the transmission on 
any one of three edges $e(7),e(9),e(11)$ has only two non-zero components,
and the vector characterizing the transmission on 
any one of eight edges
$e(1),e(2),e(3),e(4),e(5),e(6),e(8),e(10)$ has only one non-zero component.
Hence, any linear combination of the above two vectors has 
only three non-zero components at most.
Therefore, 
the vector $(1,1,1,-1)$ is not contained by  
the linear space spanned by the above two vectors.
Thus, when the message is $X_1+X_2+X_3-X_4$, 
the secrecy holds under the above attack (A).

\subsection{Problem in error detection in concrete network model}
\label{Se2F}
However, the network given in Figs. \ref{F3} and \ref{F1B}
has the problem for the detection of the error in the following meaning.
When Eve makes an active attack, 
Bob's recovering message is different from the original message due to the contamination.
Further, Bob cannot detect the existence of the error in this case.
It is natural to require the detection of the existence of the error
when the original message cannot be recovered
as well as the secrecy.
As a special attack model, we consider the following scenario 
with the attack (B1).

\begin{description}
\item[(B2)]
Our node operations are fixed to the way as Fig. \ref{F1B}.
\item[(B3)]
The message set ${\cal M}$ and all information on all edges are $\FF_2$.
\item[(B4)]
The variables $X_1, X_2,X_3,X_4$ are given as the output of the encoder.
The encoder on the source node can be chosen, but is restricted to linear.
It is allowed to use a scramble random number, which is an element of
${\cal L}:=\FF_2^k$ with a certain integer $k$. 
Formally,
the encoder is given as
as a linear function from ${\cal M} \times {\cal L}$
to $\FF_2^4$.
\item[(B5)]
The decoder on the sink node can be chosen dependently of the encoder
and independently of Eve's attack.
\end{description}
Then, 
it is impossible to make a pair of an encoder and a decoder such that
the secrecy holds and Bob can detect the existence of error.

This fact can be shown as follows.
In order to detect it, 
Alice needs to make an encoder such that 
the vector $(Y_{B,1},Y_{B,2},Y_{B,3})$ belongs to a linear subspace
because the detection can be done only by observing 
that the vector does not belongs to a certain linear subspace,
which can be written as
$\{(Y_{B,1}, Y_{B,2}, Y_{B,3})| c_1Y_{B,1}+c_2 Y_{B,2}+c_3 Y_{B,3}
=0\}$ with a non-zero vector $(c_1,c_2,c_3) \in \FF_2^3$.
That is, the encoder needs to be constructed so that 
the relation $c_1Y_{B,1}+c_2 Y_{B,2}+c_3 Y_{B,3}
=(c_1+c_2+c_3)X_1+c_1X_2+c_2X_3+c_3X_4=0$ 
holds unless Eve's injection is made.
Since our field is $\FF_2^3$,
we have three cases. 
(C1) $(c_1,c_2,c_3)$ is 
$(1,0,0)$, $(0,1,0)$, or $(0,0,1)$.
(C2) 
$(c_1,c_2,c_3)$ is 
$(1,1,0)$, $(0,1,1)$, or $(0,1,1)$.
(C3) 
$(c_1,c_2,c_3)$ is $(1,1,1)$.
If we impose another linear condition, 
the transmitted information is restricted into a one-dimensional subspace, which means that 
the message $M$ uniquely decides 
the vector $(Y_{B,1},Y_{B,2},Y_{B,3})$. 
Hence, if Eve eavesdrops one suitable variable among three variables $Y_{B,1},Y_{B,2},Y_{B,3}$, 
Eve can infer the original message.

In the first case (C1), one of three variables $Y_{B,1},Y_{B,2},Y_{B,3}$ is zero 
unless Eve's injection is made.
When $Y_{B,1}=0$, i.e., $(c_1,c_2,c_3)=(1,0,0)$,
Bob can detect an error on the edge $e(5)$ or $e(2)$
because the error on $e(5)$ or $e(2)$ affects $Y_{B,1}$ so that $Y_{B,1}$ is not zero.
However, Bob cannot detect any error on the edge $e(4)$
because the error does not affect $Y_{B,1}$.
The same fact can be applied to the case when $Y_{B,2}=0$.
When $Y_{B,3}=0$, Bob cannot detect any error on the edge $e(3)$
because the error does not affect $Y_{B,3}$.

In the second case (C2), 
two of three variables $Y_{B,1},Y_{B,2},Y_{B,3}$ have the same value unless Eve's injection is made.
When $Y_{B,1}=Y_{B,2}$, i.e., $(c_1,c_2,c_3)=(1,1,0)$,
Bob can detect an error on the edge $e(2)$ or $e(3)$
because the error on $e(2)$ or $e(3)$ affects $Y_{B,1}$ or $Y_{B,2}$ 
so that $Y_{B,1}+Y_{B,2}$ is not zero.
However, Bob cannot detect any error on the edge $e(4)$
because the error does not affect $Y_{B,1}$ nor $Y_{B,2}$.
Similarly, 
When $Y_{B,2}=Y_{B,3}$ ($Y_{B,1}=Y_{B,3}$), 
Bob cannot detect any error on the edge $e(2)$ ($e(3)$).

In the third case (C3), the relation $Y_{B,1}=Y_{B,2}+Y_{B,3}$ holds,
i.e., $(c_1,c_2,c_3)=(1,1,1)$.
Then, the linearity of the code implies that the message has the form
$a_1 Y_{B,1}+a_2 Y_{B,2}+a_3 Y_{B,3}$.
Due to the relation $Y_{B,1}=Y_{B,2}+Y_{B,3}$,
the value $
a_1 Y_{B,1}+a_2 Y_{B,2}+a_3 Y_{B,3}=(a_1 +a_2) Y_{B,2}+(a_1 +a_3) Y_{B,3}$ 
is limited to $Y_{B,1}$, $Y_{B,2}$, $Y_{B,3}$, or $0$
because our field is $\FF_2$.
Since the message is not a constant, it is limited to one of $Y_{B,1}$, $Y_{B,2}$, $Y_{B,3}$.
Hence, when it is $Y_{B,1}$, Eve can obtain the message by eavesdropping the edge $e(7)$.
In other cases, Eve can obtain the message in the same way.

To resolve this problem, we need to use this network multiple times.
Hence, in the next section, we discuss the case with multiple transmission.


\subsection{Wiretap and replacement model}\Label{MML}
In the above subsections, 
we have 
discussed the case when 
Eve injects the noise in the edges $E_A$ as well as eavesdrops the edges $E_E$.
In this subsection,
we assume that $E_A\subset E_E$
and
Eve eavesdrops the edges $E_E$ and replaces the information on the edges 
$E_A$ by other information.
While this assumption implies $m_5\le m_6$ and
the image of $\eta$ is included in the image of $\zeta_E$, 
the function $\eta$ does not necessarily equal the function $\zeta_E$
because the order that Eve sends her replaced information to the heads of edges
does not necessarily equal the order that Eve intercepts the information on the edges.
Also, this case belongs to general wiretap and addition model \eqref{E3} as follows.
Modifying the matrix $M_j$,
we define the new matrix $M_j''$ as follows.
When there is an index $i$ such that $\zeta_E(i)=j$, 
the $j+m_3$-th row vector of the new matrix $M_j''$ is defined by 
$[\delta_{j+m_3,j'}]_{1\le j'\le m_1}$ and
the remaining part of $M_j''$ is defined as the identity matrix.
Otherwise, $M_j''$ is defined to be $M_j$.
Also, we define another matrix $F$ as follows.
The $\zeta_E (i)$-th row vector of the new matrix $F$ is defined by 
$[\theta_{\zeta_E(i),j'}]_{1\le j'\le m_1}$ and
the remaining part of $F$ is defined as the identity matrix.
Hence, we have
\begin{align}
Y_{B,j}= &
\sum_{i=1}^{m_3}(M_{m_7}''\cdots M_1'')_{\zeta_B(j),i} X_i
\nonumber \\
&+
\sum_{i'=1}^{m_5}(M_{m_7}''\cdots M_1'')_{\zeta_B(j),\eta(i')} Z_{i'} 
\Label{KOG5}\\
Y_{E,j}= &
\sum_{i=1}^{m_3}(F M_{m_7}''\cdots M_1'')_{\zeta_E(j),i} X_i
\nonumber \\
&+
\sum_{i'=1}^{m_5}(F M_{m_7}''\cdots M_1'')_{\zeta_E(j),\eta(i')} Z_{i'}.
\Label{KOG6}
\end{align}
Then, we choose 
matrices $K_B'$, $K_E'$, $H_B'$, and $H_E'$ as
$K_B':=P_B  M_{m_7}''\cdots M_1'' P_{A}$, 
$K_E':=P_E  FM_{m_7}''\cdots M_1'' P_{A}$, 
$H_B':=P_B  M_{m_7}''\cdots M_1'' P_{E}^T$, 
and
$H_E':=P_E F M_{m_7}''\cdots M_1'' P_{E}^T$,
which satisfy conditions \eqref{E3} due to \eqref{KOG5} and \eqref{KOG6}.
This model ($K_B'$, $K_E'$, $H_B'$, $H_E'$) is called the {\it wiretap and replacement model}
determined by $(V,E)$ and $(E_E, E_A,\theta,\eta)$.
Notice that the projections $P_A,P_B,$ and $P_E$ are defined in Section \ref{PDG}.

Next, we discuss the strategy $\alpha'$ under the matrices $K_B'$, $K_E'$, $H_B'$, and $H_E'$ such that
the added error $Z_i$ is given as a function $\alpha_i' $ of the vector $[Y_{E,j}]_{j \in w_i}$.
Since the decision of the injected noise does not depend on the results of the decision,
we impose the causal condition defined in Definition \ref{Def3} for the subsets $w_i$.

When the relation $j \in w_i$ holds with $\zeta_E(j)=\eta(i)$,
a strategy $\alpha'$ on 
the wiretap and replacement model 
($K_B'$, $K_E'$, $H_B'$, $H_E'$)
determined by $(V,E)$ and $(E_E, \theta)$
is written by another strategy $\alpha$ on 
the wiretap and addition model
$K_B$, $K_E$, $H_B$, and $H_E$
determined by $(V,E)$ and $(E_E,\theta)$,
which is defined as
$\alpha_j([Y_{E,j'}]_{j' \in w_i}):=
\alpha_j'([\hat{Y}_{E,j'}]_{j' \in w_i})- Y_{E,j}$.
In particular, 
due to the condition \eqref{GTY},
the optimal choice $\eta_o,\{w_{o,i}\}$ 
under the partial time-ordered condition
satisfies 
the relation $j \in w_{o,i}$ holds with $\zeta_E(j)=\eta_o(i)$.
That is, under the partial time-ordered condition,
the strategy on the wiretap and replacement model 
can be written by another strategy on the wiretap and addition model.

However, if there is no synchronization among vertexes,
Eve can inject the replaced information to the head of an edge 
before the tail of the edge sends the information to the edge.
Then, the partial time-ordered condition does not hold.
In this case, the relation $j \in w_i$ does not necessarily hold with $\zeta_E(j)=\eta(i)$.
Hence, a strategy $\alpha'$ on 
the wiretap and replacement model 
($K_B'$, $K_E'$, $H_B'$, $H_E'$)
cannot be necessarily written as another strategy on 
the wiretap and addition model
($K_B$, $K_E$, $H_B$, $H_E$).

To see this fact, we discuss an example given in Section \ref{F1Ex}.
In this example, 
the network structure of the wiretap and replacement model 
is given by Fig. \ref{F3B}.

\begin{figure}[htbp]
\begin{center}
\includegraphics[scale=0.38, angle=-90]{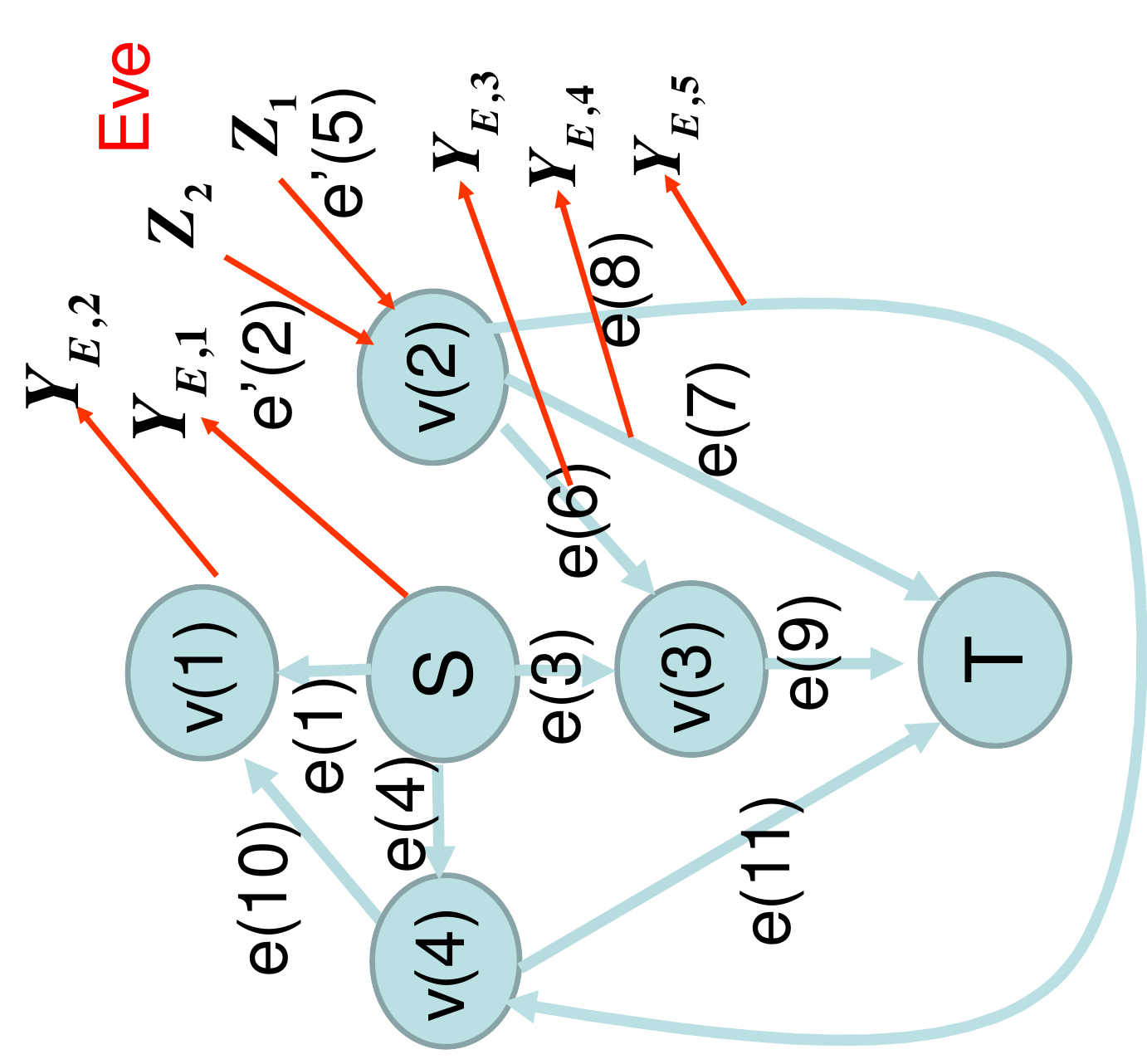}
\end{center}
\caption{Network of Section \ref{F1Ex} with wiretap and replacement model.
Eve injects the replaced information on the edges $e'(2)$ and $e'(5)$.}
\Label{F3B}
\end{figure}%


\section{Multiple transmission setting}\Label{S2-5}
\subsection{General model}\Label{S-25-1}
Now, we consider the $n$-transmission setting, where Alice uses the same network $n$ times to send a message to Bob.
Alice's input variable (Eve's added variable) is given as 
a matrix $X^n=({\bf X}_1, \ldots, {\bf X}_n) \in\FF_q^{m_{{3}} \times n}$ 
(a matrix $Z^n=({\bf Z}_1, \ldots, {\bf Z}_n) \in\FF_q^{m_{5} \times n}$),
and Bob's (Eve's) received variable is given as 
a matrix $Y_B^n=({\bf Y}_{B,1}, \ldots, {\bf Y}_{B,n})\in\FF_q^{m_{{4}} \times n}$ 
(a matrix $Y_E^n=({\bf Y}_{E,1}, \ldots, {\bf Y}_{E,n})\in\FF_q^{m_{6} \times n}$).
Then, we consider the following model as
\begin{align}
Y_B^n&=K_B X^n+ H_B Z^n, \Label{F1n} \\
Y_E^n&=K_E X^n+ H_E Z^n, \Label{F2n}
\end{align}
whose realization in a concrete network will be discussed in Sections \ref{S-25-2} and \ref{simul}.
Notice that the relations \eqref{F1n} and \eqref{F2n} with $H_E=0$
(only the relation \eqref{F1n}) were treated as the starting point of the paper
\cite{Yao2014} (the papers \cite{JLKHKM,Jaggi2008,JL}).

In this case, regarding $n$ transmissions of one channel
as $n$ different edges,
we consider the directed graph composed of $n m_5$ edges.
Then, Eve's strategy $\alpha^n$
is given as $n m_5$ functions $\{\alpha_{i,l}\}_{1\le i \le m_5, 1\le l \le n}$ 
from $Y_E^n$ to the respective components of $Z^n$.
In this case, 
we extend the uniqueness condition to the $n$-transmission version.
\begin{definition}\Label{D-U-n}
For any value of $K_E x^n$,
there uniquely exists $y^n \in \FF_q^{m_6 \times n} $ such that
\begin{align}
y^n= K_E x^n+ H_E \alpha^n(y^n).\Label{Uni2}
\end{align}
This condition is called the {\it $n$-uniqueness condition}.
\end{definition}
Since we have $n$ transmissions on each channel,
the matrix $\theta$ is given as an $(n m_1)\times (n m_1)$ matrix.
In the following,
we see how the matrix $\theta$ is given 
and how the $n$-uniqueness condition is satisfied
in a more concrete setting.

\subsection{Multiple transmission setting with sequential transmission}\Label{S-25-2}
This section discusses how the model given in Section \ref{S-25-1} can be realized in 
the case with sequential transmission as follows.
Alice sends the first information ${\bf X}_1$.
Then, 
Alice sends the second information ${\bf X}_2$.
Alice sequentially 
sends the information ${\bf X}_3, \ldots ,{\bf X}_n$.
Hence, when an injective function $\tau_E$ from $\{1,\ldots, m_1\}\times \{1,\ldots, n\}$
to $\{1,\ldots,n m_1\}$ gives the time ordering of $n m_1$ edges,
it satisfies the condition
\begin{align}
\tau_E(i,l)\le \tau_E(i',l')
\hbox{ when } i\le i' \land l\le l'.\Label{O1}
\end{align}
Here, we assume that the topology and dynamics of the network 
and the edge attacked by Eve
do not change during $n$ transmissions,
which is called the {\it stationary condition}.
All operations in intermediate nodes are linear.
Also, we assume that 
the time ordering on the network flow does not cause 
any correlation with the delayed information like Fig. \ref{F3}
unless Eve's injection is made, i.e.,
the $l$-th information ${\bf Y}_{B,l}$ received by Bob
is independent of ${\bf X}_1, \ldots, {\bf X}_{l-1},{\bf X}_{l+1}, \ldots, {\bf X}_n$,
which is called the {\it independence condition}.
The independence condition means that there is no correlation with the delayed information.
Due to the stationary and independence conditions,
the $(n m_1)\times (n m_1)$ matrix $\theta$ satisfies that
\begin{align}
\theta_{(i,l),(j,k)}= \bar{\theta}_{i,j}\delta_{k,l},\Label{O2} 
\end{align}
where $\bar{\theta}_{i,j}:=\theta_{(i,1),(j,1)}$.
When the $m_1\times m_1$ matrix $\bar{\theta} $
satisfies the partial time-ordered condition \eqref{GTY},
due to \eqref{O1} and \eqref{O2},
the $(n m_1)\times (n m_1)$ matrix $\theta$
satisfies the partial time-ordered condition \eqref{GTY}
with respect to the time ordering $\tau_E$.
Since the stationary condition guarantees that 
the edges attacked by Eve do not change during $n$ transmissions,
the above condition for $\theta$ implies 
 the model \eqref{F1n} and \eqref{F2n}.
This scenario is called the {\it $n$-sequential transmission}.

Since the independence condition is not so trivial,
it is needed to discuss when it is satisfied.
If the $l$-th transmission has no correlation with the delayed information
of the previous transmissions for $l=2, \ldots, n$,
the independence condition holds.
In order to satisfy the above independence condition,
the acyclic condition for the network graph is often imposed.
This is because any causal time ordering on the network flow does not cause
any correlation with the delayed information and achieves the max-flow if the network graph has no cycle \cite{YLCZ}.
In other words, if the network graph has a cycle,  
there is a possibility that
a good time ordering on the network flow that causes correlation with the delayed information.
However, there is no relation between the relations \eqref{F1n} and \eqref{F2n}
and the acyclic condition for the network graph,
and the relations \eqref{F1n} and \eqref{F2n} directly depend on 
the time ordering on the network flow.
That is, the acyclic condition for the network graph is not equivalent to the existence of the effect of delayed information.
Indeed, if we employ breaking cycles on intermediate nodes \cite[Example 3.1]{YLCZ}, 
even when the network graph has cycles, 
we can avoid any correlation with the delayed information\footnote{
To handle a time ordering with delayed information,
one often employs a convolution code \cite{KM}.
It is used in sequential transmission, and requires synchronization among all nodes.
Also, all the intermediate nodes are required to make a cooperative coding operation
under the control of the sender and the receiver.
if we employ breaking cycles
we do not need such synchronization as well as avoiding any correlation with the delayed information.}.
Also, see the example given in Section \ref{S81}.

To extend the causality condition, we focus on  
the domain index subsets $\{w_{i,l}\}_{1 \le i \le  m_5, 1\le l \le n}$
of $\{1, \ldots ,m_6 \}\times \{1, \ldots, n\}$
for Eve's strategy $\alpha^n=\{\alpha_{i,l}\}_{1\le i \le m_5, 1\le l \le n}$.
Then, we define the causality condition under the order function $\tau_E$. 
\begin{definition}\Label{Def3}
We say that 
the domain index subsets $\{w_{i,l}\}_{i,l}$ satisfy the {\it $n$-causal condition}
under the order function $\tau_E$ and the function $\eta$ from $\{1, \ldots, m_5\}$ to $\{ 1, \ldots,  m_1\}$
when the following two conditions hold;
\begin{description}
\item[(A1')] The relation $ H_{E;j,i} = 0$ holds for $(j,l) \notin w_{i,l}$.
\item[(A2')] 
The relation $ w_{i,l} \subseteq w_{i',l'}$ holds
when $\tau_E(\eta(i),l)\le \tau_E(\eta(i'),l')$.
\end{description}
\end{definition}

Next, we focus on the domain index subsets $\{w_{i,l}\}_{i,l}$
and 
the function $\eta$ from $\{1, \ldots, m_5\}$ to 
$\{ 1, \ldots,  m_1\}$.
We say that the pair $(\eta, \{w_{i,l}\}_{i,l})$
are {\it $n$-admissible under $\bar{\theta} $}
under the order function $\tau_E$
when $\{e(k) | k \in \im \eta \}=E_A$,
the subsets $\{w_{i,l}\}_{i,l}$ satisfy Condition (A2') for the $n$ causal condition, 
and any element $(j,l') \in w_{i,l}$ satisfies
\begin{align}
\tau_E(\zeta_E(j),l') < \gamma_{\bar{\theta}}(\eta(i),l).
\Label{LGY-n}
\end{align}
where
the function $\gamma_{\bar{\theta}}$ is defined as
\begin{align}
\gamma_{\bar{\theta}}(j,l):=
\min_{j'} \{ \tau_E (j',l)|   \bar{\theta}_{j',j}\neq 0\}.
\end{align}
Here, when no index $j'$ satisfies the condition $\bar{\theta}_{j',j}\neq 0$,
$\gamma_{\bar{\theta}}(j,l)$ is defined to be $n m_1+1$.
In the same way as Section \ref{PDG2-5},
we find that 
the $n$-admissibility of the pair $(\eta, \{w_{i,l}\}_{i,l})$
implies the $n$-causal condition under $\tau_E$ and $\eta$
for the domain index subsets $\{w_{i,l}\}_{i,l}$.

Given two $n$-admissible pairs $(\eta, \{w_{i,l}\}_{i,l})$
and $(\eta', \{w_{i,l}'\}_{i,l})$,
we say that the pair $(\eta, \{w_{i,l}\}_{i,l})$ is {\it superior to} 
$(\eta', \{w_{i,l}'\}_{i,l})$ for Eve
when $w_{{\eta'}^{-1}(j),l}'\subset w_{\eta^{-1}(j),l} $ for 
$j \in E_A$ and $l=1, \ldots , n$.
Then, we choose the bijective function $\tau_{E,\eta}$ from 
$\{1,\ldots, m_5\}\times \{1,\ldots, n\}$
to $\{1,\ldots,n m_5\}$ 
such that $\gamma_{\bar\theta} \circ\eta \circ \tau_{E,\eta}^{-1}$ is monotone increasing, where
$\gamma_{\bar\theta} \circ\eta$ is defined as
$\gamma_{\bar\theta} \circ\eta(i,l)=\gamma_{\bar\theta} (\eta(i),l)$.
The function $\tau_{E,\eta}$ expresses the order of Eve's contamination.
Then, we define $w_{\eta,i,l}:=\{ (j,l')| 
\tau_E(\zeta_E(j),l') < \gamma_{\bar{\theta}}(\eta(i),l)\}$,
which satisfies the $n$-admissibility under  $\bar{\theta}$ and  the order function $\tau_E$.

Further, when the pair $(\eta', \{w_{i,l}\}_{i,l})$ is $n$-admissible under $\bar\theta$ and $\tau_E$,
the condition \eqref{LGY-n} implies 
$w_{{\eta'}^{-1}(j),l}\subset w_{\eta,\eta^{-1}(j),l}$ for
$j \in E_A$ and $l=1, \ldots , n$,
i.e., 
$w_{\eta,i,l} $ is the largest subset under the $n$ admissibility under $\bar\theta$ and  $\tau_E$.
Hence, we obtain the optimality of $(\eta, \{w_{\eta,i,l}\}_{i,l})$
when $\bar\theta$, $\tau_E$, and $E_A$ are given.
Although the choice of $\eta$ is not unique, 
the choice of $w_{\eta,\eta^{-1}(j),l}$ for $j \in E_A$ and $l=1, \ldots , n$
is unique when $\bar\theta$, $\tau_E$, and $E_A$ are given.

In the same way as Lemma \ref{LL2}, we find that 
the $n$-causal condition with sequential transmission guarantees the 
$n$-uniqueness condition as follows.

\begin{lemma}\Label{LL2-2}
When 
a strategy $\alpha$ for the $n$-sequential transmission has domain index subsets to satisfy the $n$-causal condition,
the strategy $\alpha$ satisfies the $n$-uniqueness condition.
\end{lemma}

\begin{proof}
Consider a big graph composed of $n m_1$ edges $\{e(i,l)\}_{1\le i\le m_1, 1\le l \le n}$ and $n m_2$ vertecies $\{v(j,l)\}_{1\le j\le m_2, 1\le l \le n}$.
In this big graph, the coefficient matrix is given in \eqref{O2}.
We assign the $n m_1$ edges the number $\tau_E(i,l)$.
The $n$-causal and $n$-uniqueness conditions correspond to 
the causal and uniqueness conditions of this bog network, respectively.
Hence, Lemma \ref{LL2} implies Lemma \ref{LL2-2}.
\end{proof}

\subsection{Multiple transmission setting with simultaneous transmission}\Label{simul}
We consider anther scenario
to realize the model given in Section \ref{S-25-1}.
Usually, we employ an error correcting code for the information transmission on the edges in our graph.
For example, when the information transmission is done by wireless communication, 
an error correcting code is always applied.
Now, we assume that the same error correcting code is used on all the edges.
Then, we set the length $n$ to be the same value as the transmitted information length of the 
error correcting code.
In this case, $n$ transmissions are done simultaneously in each edge.
Each node makes the same node operation for $n$ transmissions,
which implies the condition \eqref{O2} for  
the $(n m_1)\times (n m_1)$ matrix $\theta$.
Then, the relations \eqref{F1n} and \eqref{F2n} hold 
because the delayed information does not appear.
This scenario is called the {\it $n$-simultaneous transmission}.

In fact, when we focus on the mathematical aspect, 
the $n$-simultaneous transmission can be regarded as a special case of
the $n$-sequential transmission.
In this case, the independence condition always holds even when the network has a cycle.
Further, the $n$-uniqueness condition can be derived in a simpler way without
discussing the $n$-causal condition as follows.

In this scenario, given a function $\eta$ from $\{1, \ldots, m_5\}$ to $E_A \subset \{1, \ldots, m_1\}$,
Eve chooses the added errors $(Z_{i,1}, \ldots, Z_{i,n}) \in \FF_q^n$ 
on the edge $e(\eta(i)) \in E_A$
as a function $\alpha_i $ of the vector $[Y_{E,j}]_{j \in w_i}$
with subsets $\{w_i \}_{1\le i \le m_5}$ of
$ \{1, \ldots, m_6\}$.
Hence, in the same way as the single transmission,
domain index subsets for $\alpha$ are given as
subsets $w_i\subset \{1, \ldots, m_6\}$ for $i \in \{1, \ldots, m_5\}$.
In the same way as Lemma \ref{LL2}, we have the following lemma.

\begin{lemma}\Label{LL2-3}
When 
a strategy $\alpha$ for the $n$-simultaneous transmission has domain index subsets to satisfy the causal condition, 
the strategy $\alpha$ satisfies the $n$-uniqueness condition.
\end{lemma}

In addition, the wiretap and replacement model in this setting 
can be introduced 
for the $n$-sequential transmission and the $n$-simultaneous transmission
in the same way as Section \ref{MML}.




\subsection{Non-local code and reduction theorem}
Now, we assume only the model \eqref{F1n} and \eqref{F2n} and the $n$-uniqueness condition.
Since the model \eqref{F1n} and \eqref{F2n} is given, 
we manage only the encoder in the sender and the decoder in the receiver.
Although the operations in the intermediate nodes are linear and operate only on a single transmission,
the encoder and the decoder operate across several transmissions.
Such a code is called a non-local code to distinguish operations over a single transmission.
Here, we formulate a non-local code to discuss the secrecy.
Let ${\cal M}$ and ${\cal L}$ be the message set and the set of values of the scramble random number,
which is often called the private randomness. 
Then, an encoder is given as a function $\phi_n$ from ${\cal M} \times {\cal L} $
to $\FF_q^{m_{{3}} \times n}$, and the decoder is given as $\psi_n$ from $\FF_q^{m_{{4}} \times n}$ to ${\cal M}$.
That is, the decoder does not use the scramble random number $L$ because 
it is not shared with the decoder. 
Our non-local code is the pair $(\phi_n,\psi_n)$, and is denoted by $\Phi_n$.
Then, we denote the message and  the scramble random number as $M$ and $L$. 
The cardinality of ${\cal M}$ is called the size of the code and is denoted by $|\Phi_n|$.
More generally, when we focus on a sequence $\{l_n\}$ instead of $\{n\}$,
an encoder $\phi_n$ is a function from ${\cal M} \times {\cal L} $
to $\FF_q^{m_{{3}} \times l_n}$, and the decoder $\psi_n$ is a function from $\FF_q^{m_{{4}} \times l_n}$ to ${\cal M}$.

Here, we treat $K_B,K_E,H_B$, and $H_E$ as deterministic values, and denote the pairs $(K_B,K_E)$ and $(H_B,H_E)$ by $\bm{K}$ and $\bm{H}$, respectively
while Alice and Bob might not have the full information for 
$K_E,H_B,$ and $H_E$.
Also, we assume that the matrices $\bm{K}$ and $\bm{H}$ are not changed during transmission.
In the following, we fix $\Phi_n,\bm{K},\bm{H},\alpha^n$.
As a measure of the leaked information, we adopt
the mutual information $I(M; Y_E^n,Z^n)$ between $M$ and Eve's information $Y_E^n$ and $Z^n$.
Since the variable $Z^n$ is given as a function of $Y_E^n$,
we have $I(M; Y_E^n,Z^n)=I(M; Y_E^n)$.
Since the leaked information is given as a function of 
$\Phi_n,\bm{K},\bm{H},\alpha^n$ in this situation, 
we denote it by $I(M;Y_E^n)[\Phi_n,\bm{K},\bm{H},\alpha^n]$.

\begin{definition}\Label{D-zero}
When we always choose $Z^n=0$, the attack is the same as the passive attack.
This strategy is denoted by $\alpha^n=0$.
\end{definition}

When $\bm{K},\bm{H}$ are treated as random variables independent of $M,L$,
the leaked information is given as the expectation of $I(M;Y_E^n)[\Phi_n,\bm{K},\bm{H},\alpha^n]$.
This probabilistic setting expresses the following situation. 
Eve cannot necessarily choose edges to be attacked by herself.
But she knows the positions of the attacked edges,
and chooses her strategy depending on the attacked edges.

\begin{remark}
It is better to remark that there are two kinds of formulations in network coding
even when the network has only one sender and one receiver.
Many papers \cite{Cai2002,Cai06a,Cai06,CN11,CN11b} adopt the formulation, where
the users can control the coding operation in intermediate nodes.
However, this paper adopts another formulation, in which,
the non-local coding operations are done only for the input variable $X$ and the output variable $Y_B$ like the papers \cite{KMU,Zhang,Yao2014,JLKHKM,Jaggi2008,JL}.
In contrast, all intermediate nodes make only linear operations over a single transmission, which is often called local encoding in \cite{JLKHKM,Jaggi2008,JL}.
Since the linear operations in intermediate nodes cannot be controlled by 
the sender and the receiver,
this formulation contains the case when a part of intermediate nodes do not work and output $0$ always.

In the former setting, it is often allowed to employ the private randomness in intermediate nodes.
However, we adopt the latter setting, i.e.,
no non-local coding operation is allowed in intermediate nodes, and
each intermediate node is required to make the same linear operation on each alphabet.
That is, the operations in intermediate nodes are linear and are not changed during $n$ transmissions.
The private randomness is not employed in intermediate nodes.
\end{remark}

Now, we have the following reduction theorem.
\begin{theorem}[Reduction Theorem]\Label{T1}
When the triplet $(\bm{K},\bm{H},\alpha^n)$
satisfies the uniqueness condition, 
Eve's information $Y_E^n(\alpha^n)$ with strategy $\alpha^n$ can be calculated from 
Eve's information $Y_E^n(0)$ with strategy $0$ (the passive attack),
and $Y_E^n(0)$ is also calculated from $Y_E^n(\alpha^n)$.
Hence, we have the equation
\begin{align}
I(M;Y_E^n)[\Phi_n,\bm{K},0,0]
=&
I(M;Y_E^n)[\Phi_n,\bm{K},\bm{H},0] \nonumber \\
=&
I(M;Y_E^n)[\Phi_n,\bm{K},\bm{H},\alpha^n]. \Label{NCT}
\end{align}
\end{theorem}

\begin{proof}
Since the first equation follows from the definition, 
we show the second equation.
We define two random variables $Y_E^n(0):=K_E X^n$ and 
$Y_E^n(\alpha^n):=K_E X^n+ H_E Z^n$.
Due to the uniqueness condition of $Y_E^n(\alpha^n)$, 
for each $Y_E^n(0)=K_E X^n$,
we can uniquely identify 
$Y_E^n(\alpha^n)$.
Therefore, we have 
$I(M;Y_E^n(0) )  \ge I(M;Y_E^n(\alpha^n))$.
Conversely, since $Y_E^n(0)$ is given as a function of $Y_E^n(\alpha^n)$, $Z^n$, and $H_E$,
we have the opposite inequality.
\end{proof}

\begin{remark}\Label{R5}
Theorem \ref{T1} discusses the unicast case.
It can be trivially extended to the multicast case
because we do not discuss the decoder.
It can also be extended to the multiple unicast case, 
whose network is composed of several pairs of sender and receiver.
When there are $k$ pairs in this setting, the messages $M$ and the scramble random numbers $L$ have the forms
$(M_1, \ldots, M_k)$ and $(L_1, \ldots, L_k)$.
Thus, we can apply Theorem \ref{T1} to the multiple unicast case.
The detail discussion for this extension is discussed in the paper \cite{HOKC3}.
\end{remark}

\begin{remark}
One may consider the following type of attack when Alice sends the $i$-th transmission after 
Bob receives the $i-1$-th transmission.
Eve changes the edge to be attacked in the $i$-th transmission
dependently of the information that Eve obtains in the previous $i-1$ transmissions.
Such an attack was discussed in \cite{SHIOJI} when there is no noise injection.
Theorem \ref{T1} does not consider such a situation because it assumes that Eve attacks the same edges
for each transmission. 
However, Theorem \ref{T1} can be applied to this kind of attack in the following way.
That is, we find that Eve's information with noise injection can be simulated by 
Eve's information without noise injection even when the attacked edges are changed in the above way.

To see this reduction, we consider $m$ transmissions over 
the network given by the direct graph $(V,E)$. 
We define the big graph $(V_m,E_m)$, where
$V_m:=\{(v,i)\}_{v\in V, 1\le i \le m}$ and 
$E_m:=\{(e,i)\}_{e\in E, 1\le i \le m}$
and $(v,i)$ and $(e,i)$ express the vertex $v$ and the edge $e$ on the $i$-th transmission, respectively.
Then, we can apply Theorem \ref{T1} with $n=1$ to the network given by 
the directed graph $(V_m,E_m)$
when the attacked edges are changed in the above way.
Hence, we obtain the above reduction statement
under the uniqueness condition for the network decided by 
the directed graph $(V_m,E_m)$.
\end{remark}

\begin{table*}[h!]
  \caption{Summary of security analysis}
\Label{Tb}
\begin{center}
  \begin{tabular}{|c|l|l|l|l|l|} 
\hline
Node & Eavesdropping & Vector & $\eta$ &Detection & Recovery \\
\hline
\multirow{3}{*}{$v(1)$} & $e(1):L_1$ & $(1,0,0)$&$\eta(1)=1$ &\multirow{3}{*}{$-Z_1\kappa$} &
\multirow{3}{*}{$Y_{B,2}-Y_{B,3}$}
\\
& $e(5):L_1$&$(1,0,0)$ &$\eta(2)=5$&&\\
& $e(10):L_1$ &$(1,0,0)$ &$\eta(3)=10$& &\\
\hline
\multirow{5}{*}{$v(2)$} & $e(2):L_1 \kappa+L_2(1+\kappa)+M\kappa$ & $(\kappa,1+\kappa,\kappa)$
& &\multirow{5}{*}{$Z_2-Z_1\kappa$} &
\multirow{5}{*}{$Y_{B,2}-Y_{B,3}$}
\\
& $e(5):L_1$&$(1,0,0)$ &$\eta(1)=5$&&\\
& $e(6):L_1$ &$(1,0,0)$ && &\\
& $e(7):L_1(1+\kappa)+L_2(1+\kappa)+M\kappa$ &$(1+\kappa,1+\kappa,\kappa)$ &$\eta(2)=2$&& \\
& $e(8):L_1$ &$(1,0,0)$ & &&\\
\hline
\multirow{3}{*}{$v(3)$} & $e(3):L_2+M$ &$ (0,1,1)$&{$\eta(1)=3$}& \multirow{3}{*}{$-(Z_1+Z_2+Z_3)\kappa$}& 
\multirow{3}{*}{$ (Y_{B,1}- Y_{B,3}(1+\kappa) )\kappa^{-1}$}\\
& $e(6):L_1$ &$(1,0,0)$ &$\eta(2)=6$& &\\
& $e(9):L_1+L_2+M$&$(1,1,1)$&$\eta(3)=9$&& \\
\hline
\multirow{4}{*}{$v(4)$} & 
$e(4):L_2$ & $(0,1,0)$&$\eta(1)=4$& 
\multirow{4}{*}{$-Z_1-Z_2-Z_4$} &
\multirow{4}{*}{$Y_{B,2}(1+\kappa)-Y_{B,1}$}\\
& $e(8):L_1$ &$(1,0,0)$ &$\eta(2)=8$ &&\\
& $e(10):L_1$ &$(1,0,0)$ & $\eta(3)=10$&&\\
& $e(11):L_1+L_2$&$(1,1,0)$& $\eta(4)=11$&&\\
\hline
  \end{tabular}
\end{center}
Detection expresses $Y_{B,1}-(Y_{B,3}+Y_{B,2}\kappa)$. 
If this value is not zero, Bob considers that there exists the contamination.
Recovery expresses Bob's method that decodes the message $M$ dependently of $v(i)$. 
\end{table*}

\subsection{Application to network model in Subsection \ref{F1Ex}}\Label{S81}
We consider how to apply the multiple transmission setting with sequential transmission with $n=2$
to the network given in Subsection \ref{F1Ex}, i.e., 
we discuss the network given in Figs. \ref{F3} and \ref{F1B} over the field $\FF_q$ with $n=2$.
Then, we analyze the secrecy by applying Theorem \ref{T1}.

Assume that Eve eavesdrops edges $e(2),e(5),e(6),e(7),e(8)$ 
and contaminates edges $e(2),e(5)$ as Fig. \ref{F3}.
Then, we set the function $\tau_E$ from $\{1,\ldots, 11\}\times \{1,2\}$
to $\{1,\ldots, 22\}$ as
\begin{align}
\tau_E(i,l)= i+11(l-1)\Label{tau1}.
\end{align}
Under the choice of $\eta$ given in \eqref{eta-c},
the function $\tau_{E,\eta}$ can be set in another way as
\begin{align}
\tau_{E,\eta}(i,l)= i+2(l-1)\Label{tauE}.
\end{align}
Since 
$\gamma_{\bar{\theta}}(2,1)=7$,
$\gamma_{\bar{\theta}}(5,1)=6$,
$\gamma_{\bar{\theta}}(2,2)=18$,
$\gamma_{\bar{\theta}}(5,2)=17$,
we have
\begin{align*}
w_{\eta,1,1} =&\{(1,1),(2,1)\}, ~
w_{\eta,2,1}=\{(1,1),(2,1),(3,1)\} \\
w_{\eta,1,2} =&\{(1,1),(2,1),(3,1),(4,1),(5,1),(1,2),(2,2)\} \\
w_{\eta,2,2}=&\{(1,1),(2,1),(3,1), \\
&\hspace{13ex} (4,1),(5,1),(1,2),(2,2),(3,2)\} .
\end{align*}
However, 
when the function $\tau_E$ is changed as
\begin{align}
\tau_E(i,l)&= i+5(l-1) \hbox{ for } i=1, \ldots, 5 \\
\tau_E(i,l)&= 5+i +6(l-1) \hbox{ for } i=6, \ldots, 11 ,
\Label{tau3}
\end{align}
$w_{\eta,i,l}$ has a different form as follows.
Under the choice of $\eta$ given in \eqref{eta-c},
while Eve can choose $\tau_{E,\eta}$ in the same way as \eqref{tauE},
since 
$\gamma_{\bar{\theta}}(2,1)=12$,
$\gamma_{\bar{\theta}}(5,1)=11$,
$\gamma_{\bar{\theta}}(2,2)=18$,
$\gamma_{\bar{\theta}}(5,2)=17$,
we have
\begin{align*}
w_{\eta,1,1} =&\{(1,1),(2,1),(1,2),(2,2)\} \\
w_{\eta,2,1}=&\{(1,1),(2,1),(3,1),(1,2),(2,2)\} \\
w_{\eta,1,2} =&\{(1,1),(2,1),(3,1),(4,1),(5,1),(1,2),(2,2)\} \\
w_{\eta,2,2}=&\{(1,1),(2,1),(3,1), \\
&\hspace{13ex} (4,1),(5,1),(1,2),(2,2),(3,2)\} .
\end{align*}

We construct a code, in which, 
the secrecy holds and Bob can detect the existence of the error
in this case.
For this aim, we consider two cases;
(i) There exists an element $\kappa \in \FF_q$ to satisfy
the equation $\kappa^2=\kappa+1$.
(ii) No element $\kappa \in \FF_q$ satisfies
the equation $\kappa^2=\kappa+1$.
Our code works even with $n=1$ in the case (i).
But, it requires $n=2$ in the case (ii).
For simplicity, we give our code with $n=2$ even in the case (i).

Assume the case (i).
Alice's message is $M=(M_1,M_2) \in \FF_q^2$, and
Alice prepares scramble random numbers $L_i=(L_{i,1},L_{i,2})\in \FF_q^2$ with 
$i=1,2$.
These variables are assumed to be subject to the uniform distribution independently.
She encodes them as 
$X_1= L_1$, 
$X_2= L_1 \kappa+L_2(1+\kappa)+M \kappa$,
$X_3= L_2+M$,
 and 
$X_4= L_2$.
When $Z_1=Z_2=0$, Bob receives 
\begin{align}
\begin{split}
Y_{B,1}&=X_1+X_2= L_1 (1+\kappa)+L_2(1+\kappa) +M \kappa,\\
Y_{B,2}&=X_1+X_3= L_1+L_2 +M,\\ 
Y_{B,3}&=X_1+X_4= L_1+L_2 . 
\end{split}
\Label{KLa}
\end{align}
Then, since $ M=Y_{B,2}-Y_{B,3}$,
he recovers the message by using $Y_{B,2}-Y_{B,3}$.

As shown in the following, under this code, Eve cannot obtain any information for $M$
even though she makes active attack.
Due to Theorem \ref{T1},
it is sufficient to show the secrecy when $Z_i=0$.
Eve's information is $Y_{E,1}=L_1 \kappa+L_2(1+\kappa)+M \kappa,Y_{E,2}=L_1,Y_{E,3}=L_1$, 
$Y_{E,4}= L_1 (1+\kappa)+L_2(1+\kappa)+M \kappa$, and $Y_{E,5}=L_1$.
That is, 
when variables $L_1,L_2,M$ are described by the vectors 
$(1,0,0)$, $(0,1,0)$, $(0,0,1)$, respectively,
her eavesdropping information is characterized by the vectors
$(\kappa,1+\kappa,\kappa)$, 
$(1,0,0)$, $(1,0,0)$, $(1+\kappa,1+\kappa,\kappa)$, and $(1,0,0)$ ,
and the message is by the vector $(0,0,1)$.
Since these vectors are linearly independent,
the message is independent of her eavesdropping information. 

Indeed, 
the above attack can be considered as the following.
Eve can eavesdrop all edges connected to the intermediate node $v(2)$
and contaminate all edges incoming to the intermediate node $v(2)$.
The above setting means that the intermediate node $v(2)$ is partially captured by Eve.
As other settings, we consider the case when Eve attacks another node $v(i)$ for $i=1,3,4$.
In this case, we allow a slightly stronger attack, i.e., 
Eve can eavesdrop and contaminate all edges connected to the intermediate node $v(i)$. 
That is, Eve's attack is summarized as
\begin{description}
\item[(B1')]
Eve can choose any one of nodes $v(1),\ldots, v(4)$. 
When $v(2)$ is chosen, 
she eavesdrops all edges connected to $v(2)$
and contaminates all edges incoming to $v(2)$.
When $v(i)$ is chosen for $i=1,3,4$, 
she eavesdrops and contaminates all edges connected to $v(i)$.
\end{description}
Under this attack, this code has the same secrecy
as summarized in Table \ref{Tb}.

In the case (ii),
we set $\kappa$ as the matrix
$\left(
\begin{array}{cc}
0 & 1 \\
1 & 1
\end{array}
\right)$.
Then, we introduce the algebraic extension $ \FF_q[\kappa] $ of the field $\FF_q$ by using the element $e$ to satisfy the equation 
$\kappa^2=\kappa+1$.
Then, we identify an element $(x_1,x_2)\in \FF_q^2$
with $ x_1+x_2 \kappa\in \FF_q[\kappa] $.
Hence, the multiplication of the matrix $\kappa$ in $\FF_q^2$ can be identified with 
the multiplication of $\kappa$ in $\FF_q[\kappa]$.
The above analysis works by identifying $\FF_q^2 $ with the algebraic extension $\FF_q[\kappa]$
in the case (ii).

\subsection{Error detection}\Label{S81-5}
Next, we consider another type of security,
i.e., the detectability of the existence of the error 
when $n=2$
with the assumptions (B1'), (B2) and 
the following alternative assumption;
\begin{description}
\item[(B3')]
The message set ${\cal M}$ is $\FF_q^2$, and 
all information on all edges per single use are $\FF_q$.
\item[(B4')]
The encoder on the source node can be chosen, but is restricted to linear.
It is allowed to use a scramble random number, which is an element of
${\cal L}:=\FF_q^k$ with a certain integer $k$. 
Formally,
the encoder is given as
as a linear function from ${\cal M} \times {\cal L}$
to $\FF_q^8$.
\end{description}
We employ the code given in Subsection \ref{S81}
and consider that the contamination exists when 
$Y_{B,1}-(Y_{B,3}+Y_{B,2}\kappa)$ is not zero.
This code satisfies the secrecy and the detectability as follows.

To consider the case with $v(2)$, we set $\eta(1)=5,\eta(2)=2$.
Regardless of whether Eve makes contamination, 
$Y_{B,2}-Y_{B,3}= L_1+L_2+ Z_1+M -(L_1+L_2+ Z_1)=M$.
In the following, $Y_{B,i}$ for $i=1,2,3$ expresses the variable 
when Eve makes contamination.
Hence, Bob always recovers the original message $M$. 
Therefore, this code satisfies the desired security in the case with Fig. \ref{F3}.

In the case of $v(3)$, we set $\eta(1)=3,\eta(2)=6,\eta(3)=9$.
Then, $Y_{B,1}-(Y_{B,3}+Y_{B,2}\kappa)$ is calculated to $-(Z_1+Z_2+Z_3)\kappa$.
Hence, when $Z_1+Z_2+Z_3=0$, Bob detect no error.
In this case, the contamination $Z_1$, $Z_2$, and $Z_3$ do not change $Y_{B,2}-Y_{B,3}$, i.e., do not cause any error for the decoded message.
Hence, in order to detect an error in the decoded message,
it is sufficient to check whether 
$Y_{B,1}-(Y_{B,3}+Y_{B,2}\kappa)$ is zero or not.
Since $Y_{B,2}= X_1+X_3+Z_1+Z_2+Z_3$, we have 
$M \kappa=
L_1 (1+\kappa)+L_2(1+\kappa)+M \kappa
-(L_1+L_2) (1+\kappa)
=Y_{B,1}-Y_{B,3}(1+\kappa)$.
Hence, if Bob knows that only the edges $e(3)$, $e(6)$, and $e(9)$ are contaminated,
he can recover the message by 
$(Y_{B,1}-Y_{B,3}(1+\kappa))\kappa^{-1}$.

In the case of $v(4)$, we set $\eta(1)=4,\eta(2)=8,\eta(3)=10,\eta(4)=11$.
When $
Y_{B,1}-(Y_{B,3}+Y_{B,2}\kappa)
=-(Z_1+Z_2+Z_4)=0$, Bob detects no error.
In this case, the errors $Z_1$, $Z_2$, and $Z_4$ 
do not change $Y_{B,2}-Y_{B,3}$.
Hence, it is sufficient to check whether 
$Y_{B,1}-(Y_{B,3}+Y_{B,2}\kappa)$ is zero or not.
In addition, if Bob knows that only the edges $e(4),e(8),e(10),e(11)$ 
are contaminated,
he can recover the message by 
$Y_{B,2}(1+\kappa)-Y_{B,1}$.

Similarly, in the case of $v(1)$, we set $\eta(1)=1$, $\eta(2)=5$, $\eta(3)=10$.
If Bob knows that only the edges $e(1),e(5),e(10)$ are contaminated,
he can recover the message by the original method $Y_{B,2}-Y_{B,3}$
because it equals $L_1+L_2 +M+Z_1- ( L_1+L_2 +Z_1)$.
In summary, when this type attack is done, 
Bob can detect the existence of the error.
If he identifies the attacked node $v(i)$ by another method,
he can recover the message.


\subsection{Solution of problem given in Subsection \ref{Se2F}}\Label{S82}
Next, we consider how to resolve the problem arisen in Subsection \ref{Se2F}.
That is, we discuss another type of attack given as (B1),
and study the secrecy and the detectability of 
the existence of the error under the above-explained code 
with the assumptions (B2), (B3'), (B4'), and (B5).

To discuss this problem, we divide this network into two layers.
The lower layer consists of the edges $e(7), e(9), e(11)$, which connected to the sink node.
The upper layer does of the remaining edges. 
Eve eavesdrops and contaminates any one edge among the upper layer,
and eavesdrops any one edge among the lower layer.

The vectors corresponding to the edges of the upper layer are
$(1,0,0)$,
$(\kappa,1+\kappa,\kappa)$,
$(0,1,1)$,
$(0,1,0)$.
The vectors corresponding to the edges of the lower layer are
$(1+\kappa,1+\kappa,\kappa)$,
$(1,1,1)$,
$(1,1,0)$.
Any linear combination from the upper and lower layers
is not $(0,0,1)$.
Hence, the secrecy holds under the lower type attack.
Since the contamination of this type attack
is contained in the contamination of the attack discussed in the previous subsection.
the detectability also holds.

\section{Conclusion}\Label{SCon}
We have discussed how sequential error injection affects the information leaked to Eve
when node operations are linear.
To discuss this problem, 
we have considered the possibility that the network does not have synchronization
so that the information transmission on an edges starts before 
the end of the the information transmission on the previous edge.
Hence, Eve might contaminate the information on several edges 
by using the original information of these edges.
Also, we have discussed the multiple uses of the same network
when the topology and the dynamics of the network does not changes
and there is no correlation with the delayed information.  

As a result, we have shown that 
there is no improvement by injecting an artificial noise on attacked edges.
This result can be regarded as a kind of reduction theorem
because the secrecy analysis with contamination can be reduced to that without contamination.
Indeed, when the linearity is not imposed,
there is a counterexample of this reduction theorem \cite{HOKC2}.

In addition,
we have derived the matrix formulas \eqref{F1n} and \eqref{F2n}
for the relation between the outputs of Alice and Bob and the inputs of Alice and Eve 
in the case with the multiple transmission.
As the extension of Theorem \ref{T0}, 
the similar reduction theorem (Theorem \ref{T1}) holds even for 
the multiple transmission.
In fact, as explained in Subsection \ref{S82}, this extension is essential 
because there exists an attack model over a network model 
such that the secrecy and the detectability of the error are possible 
with multiple uses of the same network 
while it is impossible with the single use of the network. 
Also, another paper will discuss the application of these results to the asymptotic setting \cite{HOKC3}.

\section*{Acknowledgments}
MH and NC are very grateful to Dr. Wangmei Guo
and Mr. Seunghoan Song
for helpful discussions and comments.

\end{document}